\documentclass[11pt]{article}
\usepackage[utf8]{inputenc}
\usepackage{arxiv_reqs}
\usepackage{tikz}
\usepackage{xspace}
\usepackage{common}

\newcommand{\fpl}{\textsc{FPL*}}

\title{Memory Bounds for the Experts Problem}
\author{
Vaidehi Srinivas\thanks{Northwestern University. 
E-mail: \email{vaidehi@u.northwestern.edu}}
\and
David P. Woodruff\thanks{Carnegie Mellon University. 
E-mail: \email{dwoodruf@cs.cmu.edu}}
\and
Ziyu Xu\thanks{Carnegie Mellon University. 
E-mail: \email{xzy@cmu.edu}}
\and
Samson Zhou\thanks{Carnegie Mellon University. 
E-mail: \email{samsonzhou@gmail.com}}
}
\date{}

\begin{document}

\begin{titlepage}
\maketitle
\thispagestyle{empty}

\begin{abstract}
Online learning with expert advice is a fundamental problem of sequential prediction.  In this problem, the algorithm has access to a set of \(n\) ``experts" who make predictions on each day. The goal on each day is to process these predictions, and make a prediction with the minimum cost. After making a prediction, the algorithm sees the actual outcome on that day, updates its state, and then moves on to the next day.  An algorithm is judged by how well it does compared to the best expert in the set.
    
The classical algorithm for this problem is the multiplicative weights algorithm, which has been well-studied in many fields since as early as the 1950s. Variations of this algorithm have been applied to and optimized for a broad range of problems, including boosting an ensemble of weak-learners in machine learning, and approximately solving linear and semi-definite programs.  However, every application, to our knowledge, relies on storing weights for every expert, and uses \(\Omega(n)\) memory.  There is little work on understanding the memory required to solve the online learning with expert advice problem (or run standard sequential prediction algorithms, such as multiplicative weights) in natural streaming models, which is especially important when the number of experts, as well as the number of days on which the experts make predictions, is large. 
    
We initiate the study of the learning with expert advice problem in the streaming setting, and show lower and upper bounds. Our lower bound for i.i.d., random order, and adversarial order streams uses a reduction to a custom-built problem using a novel masking technique, to show a smooth trade-off for regret versus memory. Our upper bounds show novel ways to run standard sequential prediction algorithms in rounds on small ``pools" of experts, thus reducing the necessary memory.  For random-order streams, we show that our upper bound is tight up to low order terms.  We hope that these results and techniques will have broad applications in online learning, and can inspire algorithms based on standard sequential prediction techniques, like multiplicative weights, for a wide range of other problems in the memory-constrained setting.
\end{abstract}

\end{titlepage}

\newpage

\section{Introduction}
The online learning with experts problem forms a general framework for sequential forecasting. 
On each day (or time), the algorithm must make a prediction about an outcome, based on the predictions of a set of ``experts.'' 
In this \emph{online learning with experts problem}, we let \(n\) be the number of experts, and the algorithm is tasked with making predictions for \(T\) days. 
For each day, the algorithm is provided with the predictions of each expert in \([n] = \{1, 2, \dots, n\}\), and then produces its own prediction based on all the information it has received in the previous days and the experts' predictions from the current day. 
The algorithm is provided with feedback on its prediction in the form of the cost of its prediction and the costs of all expert predictions on the current day. This process repeats for each day the algorithm makes predictions. 
The costs are restricted to be in the range $[0,\rho]$ for some parameter $\rho>0$.

The online learning with experts problem has been primarily studied with respect to achieving the optimal regret, i.e.,\ the additional total cost the algorithm incurs over the best performing expert in hindsight (expert that incurs the least cumulative cost), divided by the number of days. 
The well-known weighted majority algorithm was derived in \cite{littlestone_weighted_1989} for solving the discrete prediction with experts problem (where the set of possible predictions is restricted to a finite set, and the cost is 1 if the prediction is correct, and 0 otherwise) with \(O(M + \log n)\) total mistakes, where \(M\) is the number of mistakes the best expert makes. 
Optimizations to the weighted majority algorithm, such as the randomized weighted majority algorithm, achieve regret $O\left(\sqrt{\frac{\log n}{T}}\right)$. 
There have been many improvements to the weighted and randomized weighted majority algorithms in subsequent work.  There has also been work designing other algorithms that achieve this regret bound, such as follow the perturbed leader \cite{FPLKalaiVempala05}, that are less computationally expensive.  However, all of these algorithms rely on a paradigm of keeping track of the cumulative cost of every expert, which requires the algorithm to use \(\Omega(n)\) bits of memory on each day. 

In this paper, we approach the online learning with experts problem from the angle of memory-bounded learning in the data stream model, and analyze whether there exists a low-memory algorithm that still performs reasonably well when compared to the best expert. 
Specifically, we consider the memory, i.e., the space complexity, in the streaming model. 
We first note that in the discrete version of the prediction with experts problem, where there are two possible answers, $0$ and $1$, and the cost of each decision is in $\{0,1\}$ (0 for picking the correct answer, and 1 otherwise), a trivial random guessing algorithm is correct on half of the days in expectation.  
At the other extreme, the regret bound of \(O\left(\sqrt{\frac{\log n}{T}}\right)\) can be achieved by storing $O(n)$ words (i.e., the weight of each expert) in memory and implementing the weighted majority algorithm. 
Thus it is natural to ask:
\begin{quote}
\textit{What is the space/accuracy tradeoff of an algorithm for the online learning with experts problem in the streaming model?}
\end{quote}

\paragraph{Prior Work on the Experts Problem.} The experts problem has been studied in the discrete decision setting \cite{littlestone_weighted_1989} and variants where costs are determined by a loss function \cite{haussler_tight_1995,vovk_aggregating_1990,vovk_game_1998,vovk1999derandomizing,vovk2005defensive}. Consequently, many different problems can be reframed into the experts problem framework, including portfolio optimization \cite{cover1996universal,cover2011universal}, boosting \cite{freund1997decision}, and forecasting \cite{brayshaw2020new}.  
For a complete reference on these results, see \cite{cesa2006prediction}. More recent work has shown that the multiplicative weights algorithm, a generalization of the weighted majority algorithm, has a tight asymptotic regret bound of \(\sqrt{\frac{\ln n}{4T}}\) in the general case \cite{gravin_tight_2017}. 
Previous work has also addressed the computational efficiency of an experts algorithm from the perspective of time complexity. Efficient algorithms have been considered in instances where extra assumptions can be made on the expert, such as imposing a tree structure \cite{helmbold1997predicting,takimoto2001predicting}, assuming the experts are threshold functions \cite{maass1998efficient}, or assuming the experts have a nice linear structure \cite{FPLKalaiVempala05}.  Although many different algorithms have been developed in this area for these problems, they all revolve around tracking the performance of every expert, which requires at least \(\Omega(n)\) memory. To the best of our knowledge, our work is the first that examines the experts problem from a streaming perspective.

\subsection{Setup of Online Learning with Experts}

In the experts problem, there are \(n\) experts and \(T\) days. On each day \(t \in [T] = \{1, \dots, T\}\), the \(i\)-th expert makes a prediction \(x_i^{(t)}\) from the answer set \(\answers\), e.g.,\ \(\answers = \{0, 1\}\) in the discrete decision case. The algorithm then receives feedback in the form of costs \(c_i^{(t)}\) for the \(i\)-th expert on day \(t\), where we define \(c_0^{(t)}\) to be the cost incurred by the prediction of the algorithm. In the discrete decision case, the algorithm does not receive the feedback directly, but rather receives the ``correct answer'' \(y^{(t)} \in \{0, 1\}\) on day \(t\).  We can model this as a data stream that encompasses the following elements in the prescribed order: \(x_1^{(t)}, \dots, x_n^{(t)}, c_0^{(t)}, \dots, c_n^{(t)}\), where we see this sequence first for $t = 1$, then for $t = 2$, and so on until $t = T$. In the discrete decision case, for each $t$ we instead only see the sequence \(x_1^{(t)},  \dots, x_n^{(t)}, y^{(t)}\). Further, in the discrete decision case, the costs are determined by \(y^{(t)}\), and our algorithms do not need to see the sequence of the  $c_1^{(t)}, \ldots, c_n^{(t)}$, which is a useful property of our algorithms. Our lower bounds, however, still hold in the discrete decision case even if the algorithm is given $c_0^{(t)}, \ldots, c_n^{(t)}$. 

For the majority of the paper, we will consider the setup where our algorithm picks an answer made by an expert (equivalent to picking an expert). 
We only consider the discrete decision case to show that our bounds also hold when the expert explicitly picks an answer.
We use $M \in \naturals$ to denote an upper bound on the number of mistakes made by the best expert. 
Note that the streaming algorithm sees these sequences in order of increasing $t$, and cannot go back to see predictions or correct answers from previous days; the only information it has about such values is from what it stores in its memory. 
Further, our upper bounds hold in a setting where the expert predictions and feedback are considered distinct stream elements, i.e.,\ streaming memory is needed to store information about the expert predictions of the \emph{current} day (in addition to all of its information from the past days), \(x_1^{(t)}, \dots, x_n^{(t)}\),  before the costs and correct answer of the current day are seen. Note that our lower bounds hold even in the setting where \(x_1^{(t)}, \dots, x_n^{(t)}\) and \(y^{(t)}\) can be viewed at the same time. 

\paragraph{Random-Order Streaming Model.} We frequently consider the experts problem in the \emph{random-order streaming model} \cite{guha2009stream}, where we assume the days are permuted in a uniformly random order with respect to the order the days arrive in the stream. This is equivalent to saying that the input distribution over \((x^{(1)}, c^{(1)}), \dots, (x^{(T)}, c^{(T)})\) is \textit{exchangeable}, i.e.,\  any permutation on \([T]\) of the data is equally likely in the input distribution. We note that this is slightly different from the traditional notion of a random-order stream where the permutation is over elements in the stream, whereas in the experts problem, the permutation is over the pairs of expert predictions and costs that are associated with a single day in the problem. Note that the permutation only applies across the ordering of days --- the order of experts remains fixed and unchanged across the days. In addition, the exchangeability assumption in the random order model allows it to subsume the i.i.d.\ model, where \((x^{(t)}, c^{(t)}) \sim \mu\) for all \(t \in [T]\).  
As a result, all upper bounds that hold in the random order model also hold in the i.i.d.\ model.

\subsection{Our Contributions}
\label{subsec:Contributions}
We initiate the study of the online learning with experts problem in the streaming model. 
We consider streams where the order of days and corresponding expert choices and outcomes may be either (1) worst case streams, (2) random-order streams, where the order of the days is assumed to be in a random order, or (3) i.i.d. streams, where the expert choices and outcomes on each day are drawn from a fixed distribution. 
We prove upper and lower bounds for the experts problem in these different streaming models.  
We refer to a word of memory as $O(\log(nT))$ bits, and we use the $\widetilde{O}(\cdot)$ notation to suppress $\log^{O(1)}(nT)$ factors. 

\paragraph{Lower Bound for I.I.D., Random Order, and Arbitrary Streams.} 
We show that any algorithm which achieves a regret of $\delta$ with constant probability in the streaming model must use  $\Omega\left(\frac{n}{\delta^2 T}\right)$ bits of memory. This lower bound holds for both arbitrary and random-order streams. In fact, the lower bound is valid even if the correct answers and the predictions of each expert are i.i.d., which implies validity of the lower bound in random order and arbitrary-order settings.

\begin{restatable}{theorem}{thmlowerbound}
\label{thm:LowerBound}
Let $\delta$, $p<\frac{1}{2}$ be fixed constants, i.e., independent of other input parameters.  
Any algorithm that achieves $\delta$ regret for the experts problem with probability at least $1-p$ must use at least $\Omega\left(\frac{n}{\avgregret^2T}\right)$ space. 
This lower bound holds even when the costs are binary and expert predictions, as well as the correct answer, are constrained to be i.i.d.\ across the days, albeit with different distributions across the experts.
\end{restatable}

Note that for $\delta=\sqrt{\frac{1}{T}}$, Theorem~\ref{thm:LowerBound} implies that $\Omega(n)$ space is necessary. 
Thus, we should not expect to obtain the same regret bounds as the best algorithms for online learning with experts, e.g., multiplicative weights, when the space is significantly smaller than $n$. 
In other words, Theorem~\ref{thm:LowerBound} shows a separation for the online learning with experts problem between the classical centralized setting and the streaming model when only $o\left(\frac{n}{\log n}\right)$ space is permitted, in which case the \(O\left(\sqrt{\frac{\log n}{T}}\right)\) regret that is obtainable for the classical centralized setting cannot be achieved by any streaming algorithm with $o\left(\frac{n}{\log n}\right)$ space. 
Moreover, this separation holds for random-order, i.i.d., and arbitrary-order streams. 

\paragraph{Upper Bounds for Random-Order Streams.}
We next consider upper bounds for the online learning with experts problem in the random-order streaming model.
We consider the case where the costs of the decisions of each expert have value in  $[0,\rho]$ rather than in $\{0,1\}$, where $\rho>0$ is called the \emph{width} of the problem. 
\begin{theorem}[Informal]
\label{thm:general:regret:informal}
There exists an algorithm that takes a target parameter $\delta>\sqrt{\frac{16\log^2 n}{T}}$, uses $\tilde{O}\left(\frac{n}{\delta^2 T}\right)$ space, and achieves an expected regret of $\rho\delta$ in the random-order model, where $\rho$ is the width of the problem. 
\end{theorem}
When the width $\rho$ is normalized to $1$, our space dependence on the regret $\delta$ in Theorem~\ref{thm:general:regret:informal} is tight,  given the lower bound in Theorem~\ref{thm:LowerBound}.
We present the formal version of Theorem~\ref{thm:general:regret:informal} as Theorem~\ref{thm:general:regret} in Section~\ref{sec:mw}. 
Our algorithm shows that there are indeed natural tradeoffs between the memory required, the regret \(\delta\) of the streaming algorithm, the total number \(T\) of days, and the total number \(n\) of experts. 

\paragraph{Upper Bounds for Predictions on Arbitrary Streams.} 
We next consider the online learning with experts problem in the more general adversarial streaming model.
We propose an algorithm that is correct on a ``large'' fraction of days that allows for a space-accuracy tradeoff, even if the best expert makes a number of mistakes that is almost linear in $T$. 
\begin{theorem}
\label{thm:main}
(Informal) Given an upper bound $M\in\left[0,\frac{\delta^2 T}{1280\log^2 n}\right]$ on the cost that the best expert incurs, and a target regret $\delta>\sqrt{\frac{128\log^2 n}{T}}$, there exists a streaming algorithm that uses space $\widetilde{O}\left(\frac{n}{\delta T}\right)$ and with probability at least $4/5$, has regret $\rho\delta$, where $\rho$ is the width of the problem. The algorithm does not need to know $M$ in advance. 
\end{theorem}
Theorem~\ref{thm:main} provides interesting guarantees across multiple regimes of parameters for arbitrary-order streams, i.e., worst-case streams. 
First, our algorithm provides space-accuracy tradeoffs that can achieve a sublinear number of mistakes. 
Specifically, the number of mistakes made by our algorithm nearly matches the asymptotic regret bound of the multiplicative weights algorithm~\cite{gravin_tight_2017} for corresponding values of $\delta=O\left(\sqrt{\frac{\log^2 n}{T}}\right)$. 
For constant $M=O(1)$, a natural algorithm can achieve $\delta$ regret simply by iteratively choosing ``pools'' of  $\widetilde{O}\left(\frac{n}{T}\right)$ experts for ``rounds'' until the best expert is identified, where a round lasts until the pool becomes empty. Here, each expert is removed immediately after an incorrect prediction, and throughout a round, the algorithm makes a prediction by majority vote. 
We remark that the space bound of Theorem~\ref{thm:main} matches this natural algorithm in the regime where $M=O(1)$, even when the algorithm is oblivious to the value of $M$, and generalizes to handle larger values of $M$. 
In contrast, the natural algorithm uses $\tilde{O}\left(\frac{Mn}{\delta T} \right )$ space to achieve $\delta$ regret for larger values of $M$.  
In particular for constant $\delta$, our algorithm guarantees correctness on a constant $1-\delta$ fraction of days using $\widetilde{O}\left(\frac{n}{T}\right)$ space, even if the best expert is incorrect on $O\left(\frac{T}{\log^2 n}\right )$ days, i.e., the best expert makes almost a linear number of mistakes. 
Notably, this worst-case upper bound remains agnostic to the number of possible answers in the discrete decision setting, where the set of all possible answers is a finite set. 
Finally, we remark that Theorem~\ref{thm:main} uses \emph{less space} than the lower bound of Theorem~\ref{thm:LowerBound} (recall that $\delta < 1$), revealing that the hardness of Theorem~\ref{thm:LowerBound} stems from the best expert making a ``large'' number of mistakes. 

On the other hand, Theorem~\ref{thm:main} requires that the best expert incurs at most $O\left(\frac{T}{\log^2 n}\right)$ cost, even for a constant factor approximation. 
It is an interesting open question what regret bounds are achievable on arbitrary-order streams when the best expert is allowed to incur a constant fraction of errors, i.e., cost $O(T)$.  We present the formal version of Theorem~\ref{thm:main} as Theorem~\ref{thm:general-prediction-alg} in Section~\ref{sec:maj}. 

\subsection{Standard Sequential Prediction Algorithms}
At a high level, our memory-constrained algorithms for the experts problem sample subsets or ``pools" of the $n$ experts, and use standard sequential prediction techniques on these subsets.  For the sake of modularity, we will formulate this as a black-box call to a sequential prediction algorithm (without memory constraints).  In this section we define the properties that we require from such a sequential prediction algorithm, and suggest some candidate algorithms.

\begin{definition}[Sequential prediction algorithm]
We say that an algorithm \(\mathcal{A}\) is a valid \emph{sequential prediction algorithm} if, given an instance of the online learning with experts problem such that the cost of expert \(i\) on day \(t\), \(c_i^{(t)} \in [0, 1]\) for all \(i \in [n] and t \in [T]\), and a target parameter \(\eps\), we have that
\[\mathbf{E}[\text{cost of } \mathcal{A}] \le (1 + \eps) \left[\sum_{t=1}^T c_i^{(t)}\right] + \frac{\beta \ln n}{\eps},\]
for some fixed constant \(\beta\), and \(\mathcal{A}\) maintains \(O(n)\) words of memory. \label{def:seq-pred-alg}
\end{definition}

Note that the constraint that \(\mathcal{A}\) maintains \(O(n)\) words of memory is not very restrictive, as this allows \(\mathcal{A}\) to maintain the running total cost of each expert, which is all we need to run many of the best possible algorithms for this problem.

Perhaps the most well-studied such algorithm is the multiplicative weights algorithm.  Here, we introduce this algorithm in the formulation of \cite{arora2012multiplicative} (Algorithm \ref{alg:MW}).

\begin{theorem}[Theorem 2.1 in \cite{arora2012multiplicative}]
\label{thm:mw}
Suppose $c_j^{(t)}\in[0,1]$ for all $j\in[n], t\in[T]$ and $\eps\le\frac{1}{2}$. 
Then the multiplicative weights algorithm (Algorithm~\ref{alg:MW}) satisfies for each $i\in[n]$,
\[\sum_{t=1}^T\sum_{j=1}^n c_j^{(t)}p_j^{(t)}\le(1+\eps)\left[\sum_{t=1}^T c_i^{(t)}\right]+\frac{\ln n}{\eps},\]
(where the left-hand side of this inequality is the expected cost of the multiplicative weights algorithm, and the right-hand side is in terms of the cost of some particular expert.)
\end{theorem}

\begin{algorithm}[!htb]
    \caption{The multiplicative weights algorithm.}
	\label{alg:MW}
    \KwIn{Number $n$ of experts, number $T$ of rounds, parameter $\eps$}
    \begin{algorithmic}[1]
    \STATE{Initialize $w_i^{(1)}=1$ for all $i\in[n]$.}
    \FOR{$t\in[T]$}
    \STATE{$p_i^{(t)}\gets\frac{w_i^{(t)}}{\sum_{i\in[n]}w_i^{(t)}}$}
    \STATE{Follow the advice of expert $i$ with probability $p_i^{(t)}$.}
    \STATE{Let $c_i^{(t)}$ be the cost for the decision of expert $i\in[n]$.}
    \STATE{$w_i^{(t+1)}\gets w_i^{(t)}\left(1-\eps c_i^{(t)}\right)$}
    \ENDFOR
    \end{algorithmic}
\end{algorithm}

Another example of a sequential prediction algorithm is the follow the perturbed leader algorithm due to \cite{FPLKalaiVempala05} (Algorithm \ref{alg:FPL}), which maintains running totals of the cost incurred by each expert.  
On each day, the algorithm randomly ``perturbs'' the costs, and then follows the prediction of the expert with the lowest perturbed cost.   

\begin{theorem}[Theorem 1.1 in \cite{FPLKalaiVempala05} applied to the experts problem]
Suppose \(c_i^{(t)} \in [0, 1]\) for all \(i \in [n], t \in [T],\) and \(\eps \le 1\).  The \(\fpl\) (Algorithm \ref{alg:FPL}) satisfies for each \(i \in [n]\),
\[ \mathbf{E} \left[\text{cost of }\fpl (\varepsilon) \right] \le (1 + \varepsilon) \left[ \sum_{t=1}^T c_i^{(t)} \right] + \frac{8(1 + \ln n)}{\eps} .\]
For \(n \ge 3\), this satisfies the conditions of Definition \ref{def:seq-pred-alg} for \(\beta = 16\). 
\end{theorem}

\begin{algorithm}[!htb]
    \caption{The follow the perturbed leader algorithm (\(\fpl\)) from \protect\cite{FPLKalaiVempala05}, instantiated for the experts problem.}
	\label{alg:FPL}
    \KwIn{Number $n$ of experts, number $T$ of rounds, parameter $\eps$}
    \begin{algorithmic}[1]
    \FOR{$t\in[T]$}
        \FOR{$i \in [n]$}
        \STATE{Draw \(r\) from a standard exponential distribution, and set \(p_i^{(t)} = 2r/\eps\) with probability \(\frac{1}{2}\), and \(p_i^{(t)} = -2r/\eps\) otherwise}
        \ENDFOR
        \STATE{Follow the expert \(i\) for whom the sum of their total cost so far and \(p_i^{(t)}\) is the lowest}
    \ENDFOR
    \end{algorithmic}
\end{algorithm}

\subsection{Related Work on the Experts Problem}
While the space complexity of the experts problem in the models described above has not been previously studied, many related problems have been studied in the streaming model, and there are many results for the problem when space is not constrained. We discuss how the hardness, proof techniques, and algorithms for these problems relate to the space-constrained experts problem.

\paragraph{Identifying an Approximately Good Expert is Harder.} 
A closely related problem is the expert identification problem, where the algorithm must output the index of an expert that does approximately as well as the best expert at the end of the stream. 
A natural strategy might be to use a heavy-hitter algorithm to identify the best expert. 
The \emph{heavy hitters} problem is a classical problem that has been well-studied in the streaming model.  In the \(\eps\)-heavy hitters problem, the algorithm sees a stream of elements from \([n] = \{1, 2, \ldots, n\} \).  At the end of the stream, it must output any item that accounts for at least an \(\eps\) -fraction of the \(\ell_p\) norm, for some real value $p \geq 1$. 
One can consider the experts problem as a data stream of insertions to a vector \(V \in \naturals^n\), where the value at index \(i\) is the number of days the \(i\)-th expert has been correct. Running an \(\ell_p\) \(\varepsilon\)-heavy hitters algorithm over this stream would return the experts that are correct on at least \(\varepsilon \norm{V}_p\) days (see, e.g.,  \cite{CCF04,CM05,braverman2016beating,LNNT16,braverman2017bptree,bhattacharyya2018optimal} and the references therein). Another variant of the heavy hitters problem is the \(\varepsilon\)-maximum problem, which outputs a value, rather than the the index, that is within \(\varepsilon \norm{V}_1\) of the true maximum element (see, e.g., \cite{bhattacharyya2018optimal}). 
If we were to select an input where the best expert makes no mistakes, while all other experts are correct on \(T / 2 + 1\) days, \(\varepsilon\) would need to be on the order of \(O(1 / n)\) to find the best expert, and any heavy hitters algorithm for finding any good expert would require \(\Omega(n)\) space, given that their memory usages depends at least linearly on \(1/\varepsilon\). 

More generally, we note that there is a reduction from the well-studied two player set disjointness communication problem (see, e.g., \cite{chakrabarti2003near,bar2004information,weinstein2015simultaneous,braverman2013tight}), in which Alice is given a set $X\in\{0,1\}^n$ and Bob is given a set $Y\in\{0,1\}^n$ and their goal is to distinguish whether $|X\cap Y|=0$ or $|X\cap Y|=1$. 
Observe that Alice can create a stream of $n$ days so only expert $i$ can be correct on day $i$ and furthermore, expert $i$ is correct on day $i$ if and only if $X_i=1$. 
Bob similarly creates a stream of $n$ days so that the best expert on their combined stream is $X\cap Y$ if $|X\cap Y|=1$, thus allowing Bob and Alice to solve the set disjointness problem, by using another round of communication to check whether the best expert is indeed in both $X$ and $Y$. Note that
we can extend this to $T$ days
by copying Alice $T/2$ times and copying Bob $T/2$ times, and placing all copies of Alice before all copies of Bob. 
Since set disjointness requires total communication $\Omega(n)$ even for randomized protocols, this reduction immediately implies there is an \(\Omega(n)\) lower bound for identifying the best expert, even for randomized streaming algorithms and even if the best expert makes no mistakes while any other expert is correct on at most half of the total number of days.
These results show a distinction between the hardness of finding an expert that does well, relative to the best expert, and the hardness of predicting well relative to the best expert. 
We study the latter problem, and break this $\Omega(n)$ lower bound by obtaining an $\widetilde{O}(n/T)$ upper bound for the setting of parameters above. 

\paragraph{Multiplicative Weights.} Multiplicative weights is a meta-algorithm that maintains weights over a set of objects.  On iteration $i + 1$, the algorithm sets the weight of item $x$ according to the update rule $w_x^{(i + 1)} \leftarrow w_x^{(i)} \left(1 - \varepsilon P_{x}^{(i + 1)}\right)$, where $\varepsilon$ is the learning rate, and $P_x^{(i + 1)}$ is some penalty applied to $x$.
Forms of the multiplicative weights algorithm have been independently discovered for problems in many fields 
from as early as the 1950s \cite{Brown51}.  Many of these problems generalize to the discrete prediction with expert advice problem, first analyzed by Littlestone and Warmuth \cite{littlestone_weighted_1989}, or the continuous online learning with expert advice problem, and for both problems, the multiplicative weights algorithm achieves asymptotically optimal regret \cite{ordentlich-cover, cover-universal-data-compression}.  
Notable applications of the multiplicative weights algorithms include AdaBoost \cite{freund1997decision} and approximately solving zero-sum games \cite{freund1999adaptive}.  
Multiplicative weights can also be used to efficiently approximate a wide class of linear programs and semi-definite programs, which have given fast approximations for a broad range of $\textsf{NP}$-complete problems, including the traveling salesperson problem, some scheduling problems, and multi-commodity flow \cite{plotkin-fast-packing, garg-fast-packing}.
Recent work has analyzed the multiplicative weights algorithm for stochastic experts \cite{amir-corrupted-experts, rakhlin-predictable-sequences, sani-easy-data}, and bounded the regret in terms of the variance of the best expert \cite{cesa-bianchi-improved-second-order, hazan-regret-bounded-by-variance}.  There has also been much recent work in adaptively optimizing the learning rate \cite{koolen-learning-learning-rate, foster-adaptive-online-learning, chiang-gradual-variations}.  
For a more complete overview of the history and applications of multiplicative weights, the reader is referred to the surveys \cite{arora2012multiplicative, goos_-line_1998, cesa2006prediction,  foster-econ-survey}.

\paragraph{Follow the Perturbed Leader.}  The follow the perturbed leader (FPL) algorithm (Algorithm~\ref{alg:FPL}), due to Kalai and Vempala \cite{FPLKalaiVempala05}, achieves similar guarantees to multiplicative weights for the experts problem, and can be efficiently generalized to a large set of online problems.  
They define a linear generalization of these online problems.  Consider a set of possible decisions \(\mathcal{D} \subset \mathbb{R}^n\) and a set of possible events \(\mathcal{S} \subset \mathbb{R}^n\).  On each day, the algorithm chooses a decision \(d_t \in \mathcal{D}\).  Then the event of that day \(s_t \in \mathcal{S}\) is revealed, and the algorithm incurs cost \(d_t \cdot s_t\).  The total cost of the algorithm, \(\sum_{t} d_t \cdot s_t\), is evaluated against the best static decision in hindsight, \(\min_{d \in \mathcal{D}} d \cdot \sum_{t} s_t\).  We can model the standard experts problem as an instance of this problem where \(n\) is the number of experts, \(\mathcal{S}\) is made up of vectors with entries between 0 and 1, and \(\mathcal{D}\) is the set of vectors for which one index is 1, and all others are 0.  For each day \(t\), the FPL algorithm then calculates a random perturbation vector \(p_t \in \mathbb{R}^n\). 
Then it follows the decision that optimizes $\min_{d \in \mathcal{D}} \quad d \cdot \left(p_t + \sum_{i \le t} s_i \right)$. 
In the case of the standard experts problem, this is generating a perturbation for each expert, and then following the expert for whom the sum of their cost and their perturbation is the smallest.

This linear generalization means that for some structured problems, this algorithm is computationally more efficient than multiplicative weights.
However, like multiplicative weights, it requires access to the running cost \(\sum_{t} s_t\).  For the general experts problem, this is an \(n\) dimensional vector, so this still requires \(\Omega(n)\) memory.

\paragraph{Online Convex Optimization.} 
A common setting in online convex optimization is to minimize the regret, defined by $\sum_{t=1}^T f_t(x_t)-\min_{x\in X}\sum_{t=1}^T f_t(x)$, where $X$ is some convex set and the functions $f_1,\ldots,f_T:X\to\mathbb{R}$ are convex cost functions. 
A special case of online convex optimization is when the goal is to minimize a convex function $f$ over a convex domain $X$. 
If the algorithm is given oracle access to a noisy gradient, i.e., an oracle that outputs an unbiased estimate to the gradient with ``small'' variance, then stochastic gradient descent is known to have expected regret at most $O(1/\sqrt{T})$. 
For a more precise statement, see~\cite{Shamir16,Hazan16,TaiSBV18}. 

\paragraph{Multi-Armed Bandits.} 
Space complexity has been considered in the related problem of \emph{multi-armed bandits}. 
The multi-armed bandits problem is a classic problem in reinforcement learning, in which there are some number of \emph{arms} and each arm has a fixed reward distribution that can be sampled from at each time step.  
\cite{liau2018stochastic} has shown that only constant space is required to achieve regret that is within an \(O(1 / \Delta)\) factor of the optimal regret, where \(\Delta\) is the difference between the mean reward of the best and second best arms. This space efficient algorithm for the bandits problem, however, is not applicable to the experts problem, since it does not capture the adversarial nature of sequential prediction, i.e.,\ the performance of an expert changes on a daily basis, while the reward distribution of each arm is static in the standard streaming model. Nor does it leverage the ability of the algorithm to view the results of all experts on each day. Hence, expert algorithms are not comparable to bandit algorithms, since an expert algorithm has information about all ``arms'' on each day. \cite{assadi_exploration_2020} analyze the problem of finding the best arm with optimal sample complexity in the streaming arm model, where the algorithm must save an arm to sample from it. They prove a tight bound of requiring \(\Theta(k)\) arms of space to find the top \(k\) arms. These results show that solving the experts problem is fundamentally harder than the multi-armed bandits problem, since solving the bandits problem with low regret can be done in constant space.

\paragraph{Learning in Streams.} There has been a substantial amount of work that analyzes the tradeoff between space and sample complexity for statistical learning and estimation problems in the streaming model, where the stream elements are assumed to be samples drawn i.i.d.\ from a fixed distribution. 
A series of work has studied the problem of inferring the index of a row sampled from a matrix \cite{garg2017extractor,raz2017time,garg2019time}, and the parity learning problem \cite{raz2018fast,Kol2017}.
More recent work has also analyzed distribution testing relevant to cryptographic settings, such as lower bounds in the streaming model for testing against Goldreich's pseudorandom generator \cite{garg2020time}. 
Other lines of work examine more specific learning problems in the streaming model, such as finding correlations in multivariate data \cite{Dagan2018}, collision probability estimation, finding the connectivity of an undirected graph, and rank estimation \cite{crouch2016}. In the distributed setting, communication lower bounds have been analyzed for convex optimization \cite{Arjevani2015}. 
Our study of the experts problem, however, makes no assumptions on the distributions each element in the stream is drawn from, and is a prediction rather than an inference problem.
Space usage was considered by \cite{kar_generalization_2013} in their analysis of pairwise losses in online learning, but not in the general sense of space complexity. They consider a modified form of regret under the pairwise loss with respect to a finite buffer of previous stream items rather than all previous items in the stream. This problem, however, is very different from the experts setting of proving space complexity bounds for any algorithm, while still comparing the performance of the algorithm to that of the best expert. 

\subsection{Overview of our Techniques}
\label{sec:overview}

\subsubsection{Lower Bound by Reducing to Distributed Detection}
We first give an overview of the lower bounds that we present in Section~\ref{sec:lb}. 
We create a new problem that combines \(n\) instances of the distributed detection problem~\cite{braverman2016communication}, with each instance corresponding to an expert --- we call this new problem \diffdist. 
In essence, \diffdist\ is the problem of distinguishing between:
\begin{enumerate}
    \item Every expert flips an independent fair coin to determine its prediction, i.e., each expert predicts correctly with probability \(\frac{1}{2}\).
    \item A single expert predicts correctly with probability \(\frac{1}{2} + O(\avgregret)\) on each day, and every other expert predicts correctly with probability \(\frac{1}{2}\).
\end{enumerate}
We note that the predictions of each expert form a separate instance of the distributed detection problem, each of which has a randomized communication lower bound of \(\Omega(\frac{1}{\avgregret^2})\). 
We then use a careful combination of existing techniques, e.g.,~\cite{zhang_informationtheoretic_lower_2013,braverman2016communication,GargMN14,Raginsky16} to show an \(\Omega\left(\frac{n}{\delta^2}\right)\) randomized communication lower bound for the \diffdist\ problem.  

After having shown a randomized communication lower bound for the \diffdist\ problem, one of our key ideas is then to introduce a randomized reduction from the $\diffdist$ problem, so that each player corresponds to a day, and each expert prediction is the corresponding bit in the \diffdist\ problem. 
We would like to say that the single expert with higher probability of correctness translates to a separation in the cost of the algorithm for online learning with expert advice. 
However, even if the experts are incorrect, it seems possible that an algorithm could ignore the experts and still have high accuracy. 
For example, if we let the the correct answer be 1 for every day, an experts algorithm could be perfectly accurate in case (1) simply by predicting 1 everyday.
Thus, we create a mask for each day by setting it to be a random fair coin flip and we XOR both the outcome of each day, as well as the output of each expert, by the mask. 
Hence in case (1), the algorithm cannot have accuracy significantly higher than \(\frac{1}{2}\) with constant probability, regardless of its output sequence, since the expert predictions defined via our masking, and the correct answer, all remain independent fair coin flips. 
On the other hand, in case (2), the algorithm must be correct on a good fraction of days to keep up with the best expert. 
This results in a separation in performance between cases (1) and (2). 
Our hard instance is inherently distributional, and the choice of hard distribution is crucial to ensure that the algorithm has no information in case (1) about the correct answer on a future day, regardless of the past. 

\subsubsection{Upper Bound for Random-Order Streams}
We first consider upper bounds for the online learning with experts problem in random-order streams, corresponding to the results in Section~\ref{sec:mw}.  
For simplicity, we will describe an overview of our algorithms and proof techniques in the setting where the costs are restricted to \(\{0, 1\}\), i.e., an expert either makes a mistake or does not. 
Due to space constraints on the algorithm, we can only afford to sample a small number of experts in each round. 
Thus our algorithm (in Algorithm~\ref{alg:SublinearAlgorithm:MW}) initializes a pool of $k=O\left(\frac{n\log n}{\delta^2 T}\right)$ different experts from $[n]$ at the beginning of each round. We then run a standard sequential prediction algorithm on this restricted pool of experts. 
If this pool of experts makes too many total mistakes {\it in expectation}, then we resample a new pool of experts and run the sequential prediction algorithm on this new pool.
Note that we can explicitly compute the expectation of the pool, since we have access to the expert predictions in the pool, their corresponding weights, and the outcomes across each day over the duration of the round. 

In the random-order model, we can show that if the best expert does not make too many mistakes, then the sequential prediction algorithm will perform well upon sampling the best expert. 
Thus, to prove our algorithm has low expected regret, we must demonstrate that one of two cases must be true: (1) the algorithm does not sample the best expert because the algorithm has already been performing well or (2) the algorithm samples the best expert ``early'' enough in the stream and the best expert is not subsequently discarded. 
To handle the first case, we observe that we only delete the pool of experts if the sequential prediction algorithm is performing poorly, so if the total number of rounds is low, there must be rounds with significantly long duration and also sufficiently high expectation. 
Thus the algorithm will perform well in expectation. 

\paragraph{Analysis of Pool Selection Times.} 
However, the second case foreshadows an issue in the analysis: if the best expert is never discarded from the pool, then the best expert can only be added to the pool in the last round. 
Thus, even if we condition on the entire algorithm using $R$ rounds for some integer $R>0$, the probability that the best expert is not added to the pool in round $R$ may be significantly larger than the probability that the best expert is added to the pool in round $R$. 
This issue is further compounded by the fact that once a pool is selected, the day on which the next round begins is completely deterministic and possibly adversarial, so it does not suffice to, for example, consider the probability the best expert is added to the pool on a random day. 

We overcome this challenge by ``decoupling'' the number of rounds from the sampling of the best expert. 
What we mean by this is as follows. We consider the distribution of all days on which new rounds begin. 
We can simulate the sampling process with a sequence of times $t_1,t_2,\ldots$ so that each $t_i$ is drawn from the distribution of possible times that round $i$ can end, conditioned on the entire history of the process up to time $t_{i-1}$. 
Observe that this is a well-defined sequential process for defining each term $t_i$, in the sense that to obtain $t_i$,  we can simply draw from a distribution of possible durations for a round and then add the duration to $t_{i-1}$. 

We observe that due to the distributional properties of the random-order model, once the algorithm samples the best expert, then with high probability the sequence of rounds will terminate. 
Moreover, since the algorithm performs sampling with replacement between rounds, the probability that the best expert is added to the pool on each drawing of $k$ experts is the same across all rounds. 
Thus, the probability distribution for the total number of rounds can be related to a geometric distribution --- if the algorithm uses $R$ rounds, then the best expert cannot be sampled in the first $R-1$ rounds with high probability. 
This allows us to show that if the algorithm samples the best expert, there was likely a small number of rounds and therefore the total cost of the algorithm is not too high. 

\paragraph{Unknown Error for the Best Expert.}
It remains to remove the assumption of knowing the error rate for the best expert, for which we again use the promise of the random order streaming model, which allows us to use short prefixes of days in order to obtain an estimate of the error rate.  

To that end, we note that it suffices to acquire a $(1+O(\delta))$-approximation to the number of mistakes made by the best expert, since the regret will only be increased by $O(\delta)$ if we have such an estimate. 
To find a $(1+O(\delta))$-approximation to the number of mistakes of the best expert, we initialize our guess $\gamma$ for the mistakes to be $\frac{T}{2}$. 
We then split the stream into epochs of length $O\left(\frac{\delta T}{\log\frac{1}{\delta}}\right)$ and perform a binary search by repeatedly updating $\gamma$ depending on whether the current guess is too high or too low based on the performance of the best expert in the epoch. 
Thus by epoch $k$, our guess $\gamma$ is within a $\left(1+\frac{1}{2^k}\right)$-factor of the actual number of mistakes made by the best expert. 
Hence it suffices to use $O\left(\log\frac{1}{\delta}\right)$ epochs to update $\gamma$, which can only increase the total regret additively by $\delta$, since each epoch has length $O\left(\frac{\delta T}{\log\frac{1}{\delta}}\right)$. 

\subsubsection{Upper Bound for Arbitrary-Order Streams.}
We consider arbitrary order streams in Section~\ref{sec:maj}. 
Unfortunately, when the stream no longer arrives in a random order, then we again have no guarantees on how the best expert will perform if it is sampled at any given time. 
We observe that if the costs are $\{0,1\}$ for each day, then we can attempt to emulate the simpler majority elimination algorithm by removing all incorrect experts on a day on which the algorithm is incorrect. 
Thus, our starting point is an algorithm that initializes a pool of $k=O\left(\frac{n\log n}{T\delta}\right)$ different experts from $[n]$ at the beginning of each round and removes incorrect experts on incorrect days until the pool is depleted, at which point the next round begins and a new pool of $k$ different experts from $[n]$ is initialized. 
On each day, the algorithm outputs the majority vote of the experts in the pool. 

However, removing all incorrect experts would significantly increase the chance that the best expert is removed from the pool, even over multiple rounds. 
For example, if the best expert makes a constant number of mistakes, then it is possible that it only survives a constant number of days before it is removed from the pool. 
If all other experts perform poorly, then the algorithm could only be correct on a constant number of days in every group of $O\left(\frac{T\delta}{\log n}\right)$ days, which is subconstant even if $\delta$ is a constant. 
Therefore, we should relax the conditions for removal of experts; a natural choice is to only remove experts that have been incorrect for $\frac{\delta}{4}$ fraction of the time since the pool has been initialized, \emph{regardless of the outcome of each day}. 
The intuition is that all experts make errors on at most  a $\frac{\delta}{4}$ fraction of the days in the pool, so the algorithm should make errors on at most an $O(\delta)$ fraction of the days over the pool. 

\paragraph{Accumulation of Errors.}
However, this surprisingly fails because it allows experts to ``build-up'' future errors by having good accuracy on previous days. 
For example, suppose we have a pool of $100$ experts and we choose to eliminate experts that are wrong on half of the days since the pool has been initialized. 
Suppose all experts are correct on the first $50$ days but then from day $51$ to day $100$, exactly half the experts are wrong on every single day. 
On day $100$, half of the experts are eliminated and the algorithm has made $50$ mistakes, but the remaining experts have not made any mistakes. 
Thus even if half of the remaining $50$ experts are wrong on \emph{every single day} from day $101$ to day $200$, they will not be eliminated until day $200$, which causes the algorithm to err on every single day during that interval. 
We can continue this geometric approach by allowing half of the experts to be wrong on an interval with double the length, e.g., $13$ of the remaining $25$ experts are wrong every single day from day $200$ to day $400$, so that the algorithm will always be incorrect after the first $50$ days, which clearly contradicts the desired claim. 
See Figure~\ref{fig:buildup} for an example of such an accumulation. 

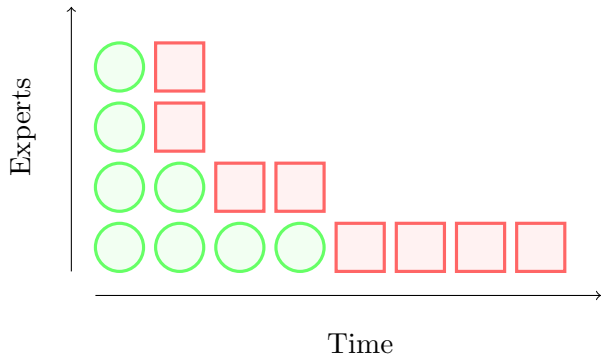
\begin{figure*}[!htb]
\centering
\begin{tikzpicture}[scale=0.8]
\filldraw[color=green!60, fill=green!5, very thick](0,3) circle (0.4);
\filldraw[color=red!60, fill=red!5, very thick](1-0.4,3-0.4) rectangle +(0.8,0.8);

\filldraw[color=green!60, fill=green!5, very thick](0,2) circle (0.4);
\filldraw[color=red!60, fill=red!5, very thick](1-0.4,2-0.4) rectangle +(0.8,0.8);

\filldraw[color=green!60, fill=green!5, very thick](0,1) circle (0.4);
\filldraw[color=green!60, fill=green!5, very thick](1,1) circle (0.4);
\filldraw[color=red!60, fill=red!5, very thick](2-0.4,1-0.4) rectangle +(0.8,0.8);
\filldraw[color=red!60, fill=red!5, very thick](3-0.4,1-0.4) rectangle +(0.8,0.8);

\filldraw[color=green!60, fill=green!5, very thick](0,0) circle (0.4);
\filldraw[color=green!60, fill=green!5, very thick](1,0) circle (0.4);
\filldraw[color=green!60, fill=green!5, very thick](2,0) circle (0.4);
\filldraw[color=green!60, fill=green!5, very thick](3,0) circle (0.4);
\filldraw[color=red!60, fill=red!5, very thick](4-0.4,0-0.4) rectangle +(0.8,0.8);
\filldraw[color=red!60, fill=red!5, very thick](5-0.4,0-0.4) rectangle +(0.8,0.8);
\filldraw[color=red!60, fill=red!5, very thick](6-0.4,0-0.4) rectangle +(0.8,0.8);
\filldraw[color=red!60, fill=red!5, very thick](7-0.4,0-0.4) rectangle +(0.8,0.8);

\draw[->] (-0.4,-0.8) -- (8,-0.8);
\node at (4,-1.6){Time};
\draw[->] (-0.8,-0.4) -- (-0.8,4);
\node[rotate=90] at (-1.6,2){Experts};
\end{tikzpicture}
\caption{A simple example where for $\delta=\frac{1}{2}$, removing experts at least $1-\delta$ error rate can still lead to a pool with $T-\frac{T}{k}$ errors. In this case, green circles denote correct predictions by experts and red squares denote incorrect predictions.}
\label{fig:buildup}
\end{figure*}

The key to the above counterexample is that experts that are incorrect on later days can cause a larger number of incorrect outputs by the algorithm because these experts were correct on previous intervals. 
At a first glance, it seems we can avoid this issue by instead resetting a timer for the remaining experts each time the size of the pool roughly halves. 
Namely, suppose we define a timer $u$ to first demarcate the beginning of the round. 
Any expert that is inaccurate for at least $\frac{\delta}{4}$ fraction of the days since time $u$ is deleted. 
Each time the size of $P$ decreases by roughly half, the variable $u$ is updated to the new time. 
As before, the current round ends when the pool is completely depleted of experts, at which point the next round begins and a new pool is chosen. 
However, this \emph{still} does not work because now the timer can be set adversarially to always cause the best expert to be deleted. 

Surprisingly, the issue is alleviated if we instead require an even more stringent demand from the experts in the pool. 
Instead of asking for experts to make errors on at most $\frac{\delta}{4}$ fraction of the days in the pool, we instead ask experts to make errors on at most $O\left(\frac{\delta}{\log n}\right)$ fraction of the days in the pool. 
Since the timers are no longer reset, it is once again possible for the best expert to not be deleted. 
Moreover, the extra $O(\log n)$ factor allows us to overcome to build-up of errors in the previous counterexample, because the errors can only accumulate over $O(\log n)$ rounds. 

The intuition for the algorithm is that one of two cases should hold. 
Either there is a small number of rounds, which indicates that the experts in some pool performed well over a large period of time, or there is a large number of rounds, in which case it is likely that the best expert is added to the pool in some round. 
We would like to show that in the latter case, the best expert being added to the pool compels the algorithm to perform well overall. 
The idea is that if we add the best expert to the pool on a \emph{random} day, it is unlikely that the best expert will ever be deleted from the pool and thus the algorithm will have good accuracy. 

\paragraph{Decoupling for Arbitrary-Order Streams.}
Whereas the decoupling argument in the random-order model crucially used the fact that the best expert would not be deleted if it was sampled to the pool, this property no longer holds for arbitrary-order streams. 
For instance, there could be $\Omega(M)$ consecutive days in which the best expert makes a mistake, so that our algorithm will likely delete the best expert during that time. 
Thus, we define ``bad'' times as days on which the best expert would be deleted from the pool if it were sampled on that day. 
Then we can upper bound the total number of rounds by the sum of the number of rounds starting on bad times and the number of rounds starting on good times. 
Since the number of rounds beginning on bad times is upper bounded by the number of bad times, it suffices to upper bound the probability distribution for the number of rounds starting on good days. 
However, if the best expert is sampled on a good day, then it will never be deleted, which terminates the sequence of resamplings. 
Thus, the probability distribution for the number of rounds initiated on good days is upper bounded by a geometric distribution. 
Hence, we can show that the number of rounds initiated on good days and thus the \emph{total} number of rounds is ``low'' with high probability. 
It follows that with high probability, the total number of mistakes by the algorithm must therefore also be low since the algorithm only resamples if its accuracy is poor. 
Thus although these techniques may not be as modular as the random-order streams, our algorithm can achieve high-probability bounds for arbitrary-order streams.

\subsection{Technical Preliminaries}

\subsubsection{Standard Technical Tools}
We use the following concentration inequalities.
\begin{fact}[Hoeffding's inequality \cite{hoeffding_probability_1963}]
Suppose \(X_1, \dots, X_n\) are independent random variables in \([a, b]\). 
Then, for any \(t > 0\),
\[\PPr{|\overline{X} - \expect[\overline{X}]|\ge t}\le \exp\left(-\frac{2 n t^2}{(b - a)^2}\right),\]
where \(\overline{X}=\frac{1}{n}\sum_{i=1}^s X_i\). Thus, for a constant probability \(\alpha \in (0, 1)\):
\begin{align*}
    \prob{|\overline{X} - \expect[\overline{X}]| \geq \sqrt{\frac{\log(1 / \alpha)(b - a)^2}{2n}} } \leq \alpha.
\end{align*}
\end{fact}

\begin{fact}[Multiplicative Chernoff bound]
\cite{mitzenmacher_probability_2005}
	Let \(X_i\) where \(i \in [n]\) be independent random variables taking values in \(\left\{0, 1\right\}\). Let \(X\) denote their sum. Then, the following tail bounds hold:
	\begin{align*}
	\prob{X > (1 + \delta)\expect\left[X\right]} &\leq \exp\left(-\frac{\delta^2\expect\left[X\right]}{2 + \delta}\right) & \text{for}\ & 0 \leq \delta\\
	\prob{X < (1 - \delta)\expect\left[X\right]} & \leq \exp\left(-\frac{\delta^2 \expect\left[X\right]}{2}\right) & \text{for}\ & 0 \leq \delta \leq 1.
	\end{align*}
\end{fact}

We also need the two following distances on distributions.
\begin{definition}[Distances between distributions]

The \textbf{squared Hellinger distance} between two discrete distributions \(P, Q\) with \(K\) outcomes is defined as \(h^2(P, Q) = \frac{1}{2}\sum_{i = 1}^K (\sqrt{P_i}-  \sqrt{Q_i})^2\).   

The \textbf{total variation distance} is defined as \(\mathrm{TV}(P, Q) =\frac{1}{2} \sum_{i = 1}^K |P_i - Q_i|\).  

Further, for any discrete distributions \(P, Q\), the following inequality is true: \(h^2(P, Q) \geq \frac{1}{2}\mathrm{TV}(P, Q)^2\).
\end{definition}

\subsubsection{Information Theory and Communication Complexity}

We use standard notions of \(H(X)\) for the entropy of random variable \(X\), and \(H(X \mid Y)\) for the conditional entropy of random variable \(X\) conditioned on random variable \(Y\). Mutual information is defined as \(I(X;Y) = H(X) - H(X \mid Y)\). For a more detailed reference, see \cite{cover1999elements}. 

We prove our lower bounds using the \emph{blackboard model}, where each of \(T\) parties communicates by posting a message to the blackboard, and we denote the transcript of all communication in a protocol as \(\transcript \in \left\{0, 1\right\}^*\), which has access to a public source of randomness. The communication complexity of a protocol, as a result, is the bit length of the transcript \(\left|\transcript\right|\). 
We utilize a lower bound on the following problem to prove a streaming lower bound for the experts problem.

\begin{definition}[Distributed detection problem \cite{braverman2016communication}]
	For fixed distributions $\mu_0$ and $\mu_1$, let $X^{(1)}, \dots, X^{(T)}$ be sampled i.i.d.\ from $\mu_V$, for $V \in \{0, 1\}$.  The \textbf{distributed detection problem} is the task of determining whether $V = 0$ or $V = 1$, given the values of $X = (X^{(1)}, \dots, X^{(T)})$. 
\end{definition}

Consider a specific instance of the distributed detection problem where \(\mu_0 = B_{0.5}\) and \(\mu_1 = B_{0.5 + \varepsilon}\) and \(\varepsilon \in [0, \frac{1}{2}]\). We will refer to this as the \textit{\(\varepsilon\)-distributed detection problem}. We would like to show an mutual information lower bound on any algorithm that solves the \(\varepsilon\)-distributed detection problem with high probability. To do so, we rely on the framework for deriving a strong data processing inequality (SDPI) introduced in \cite{braverman2016communication}. They first define a SDPI constant, \(\beta\).
Let \(B_{p}\) denote the Bernoulli distribution with parameter \(p \in [0, 1]\). 
\begin{definition}[SDPI constant \cite{braverman2016communication}]
Let \(V \sim B_{0.5}\) (fair coin) and the channel \(V\rightarrow X\) be defined as it is in the distributed detection problem above. Then, there exists a constant \(\beta \leq 1\) that depends on \(\mu_0\) and \(\mu_1\), s.t. for any transcript \(\Pi\) that depends only on \(X\) i.e.\ \(V \rightarrow X \rightarrow \Pi\) is a Markov chain, we have
\begin{align}
    I(V; \Pi) \leq \beta I(X ; \Pi).
    \label{eqn:SDPIBound}
\end{align} Let \(\beta(\mu_0, \mu_1)\), the SDPI constant, be the infimum over all possible \(\beta\) such that \eqref{eqn:SDPIBound} holds.  
\end{definition}

We can use the SDPI constant, \(\beta(\mu_0, \mu_1)\), to relate the squared Hellinger distance between the distribution of \(\Pi\) under \(\mu_0\) vs.\ \(\mu_1\) to the mutual information between the inputs \(X\) and \(\Pi\).
\begin{fact}[Theorem 3.1 of \cite{braverman2016communication}]
Suppose \(\mu_1 \leq c \cdot \mu_0\). Then, in the distributed detection problem, we have that
\begin{align*}
\mathrm{h}^2(\Pi\mid_{V = 0}, \Pi\mid_{V = 1}) \leq K(c + 1) \beta(\mu_0, \mu_1) \cdot I(X;\Pi \mid V= 0),
\end{align*} where \(K\) is a universal constant, and \(\Pi \mid_{V = v}\) denotes the distribution of \(\Pi\) conditioned on \(V = v\).
\label{fact:HellingerBound}
\end{fact}

Our last ingredient is fact that implies an upper bound on \(\beta(\mu_0, \mu_1)\). 
\begin{fact}[Lemma 7 of \cite{zhang_informationtheoretic_lower_2013}]
    Assume we sample \(V\) from a distribution, and suppose \\
    \(\sup_x\frac{\mathbf{Pr}_{X_i\sim\mu_1}(X_i=x_i)}{\mathbf{Pr}_{X_i\sim\mu_0}(X_i=x_i)}\le c\) for all \(i \in [T]\).
    
    Then, the following inequality holds:
    \begin{align*}
		I(V;\transcript) \leq  2(c^2 - 1)^2 I(X;\transcript).
	\end{align*}
	\label{fact:SDPI}
\end{fact} 
\vspace{-22pt}
This immediately implies the following corollary about the \(\varepsilon\)-distributed detection problem.
\begin{corollary}
Let \(\varepsilon \in [0, 0.5]\), \(\mu_0 = B_{0.5}\), and \(\mu_1 = B_{0.5 + \varepsilon}\). Setting \(b = 1 + \varepsilon\), \Cref{fact:SDPI} implies the following inequality:
\begin{align*}
    I(V;\Pi) \leq 2(2\varepsilon + \varepsilon^2)^2I(X;\Pi) = (8\varepsilon^2 + 8\varepsilon^3 + 2\varepsilon^4)I(X;\Pi) \leq 18\varepsilon^2I(X;\Pi).
\end{align*} Consequently, \(\beta(\mu_0, \mu_1) \leq 18\varepsilon^2\).
\label{corollary:SDPIBetaBound}
\end{corollary}

Given the above results, we can now show a lower bound for the mutual information between \(\Pi\) and \(X\) in the case where \(V = 0\).
\begin{theorem}[Mutual information lower bound for the \(\varepsilon\)-distributed detection problem]
Let \(\Pi\) be the transcript of an algorithm that solves \(\varepsilon\)-distributed detection problem with probability \(1 - p\) for any fixed choice of \(p \in [0, 0.5)\). Then, we can provide the following lower bound:
\begin{align*}
     I(X;\Pi \mid V = 0) = \Omega\left(\varepsilon^{-2}\right).
\end{align*}
\label{thm:EpsilonIC}
\end{theorem}
\vspace{-22pt}
\begin{proof}
By the definition of \(\Pi\), we can apply some function \(\widehat{V}\) such that \(\widehat{V}(\Pi) = V\) with a constant probability strictly greater than 0.5. Thus, \(\mathrm{TV}(\widehat{V}(\Pi)\mid_{V = 0}, \widehat(\Pi)\mid_{V = 1}) = \Omega(1)\). This implies the following:
\begin{align}
\mathrm{h}^2(\Pi\mid_{V = 0}, \Pi\mid_{V = 1}) \geq \mathrm{h}^2(\widehat{V}(\Pi)\mid_{V = 0}, \widehat{V}(\Pi)\mid_{V = 1}) \geq \frac{1}{2}\mathrm{TV}(\widehat{V}(\Pi)\mid_{V = 0}, \widehat{V}(\Pi)\mid_{V = 1})^2 \geq  \Omega(1).
\label{eq:CorrectnessBound}
\end{align}
The first inequality is by data processing, the second is by the relationship between squared Hellinger and total variation distance, and the last inequality is by correctness of \(\Pi\).

We note that we can apply \Cref{fact:HellingerBound}, with \(c = 2\), to derive the following:
\begin{align*}
\mathrm{h}^2(\Pi\mid_{V = 0}, \Pi\mid_{V = 1}) \leq 3K \beta(\mu_0, \mu_1) \cdot I(X;\Pi \mid V = 0).
\end{align*} 

Applying \eqref{eq:CorrectnessBound} and \Cref{corollary:SDPIBetaBound} to the above equation, we get:
\begin{align*}
\Omega(1) \leq 3K(18\varepsilon^2) \cdot I(X;\Pi \mid V = 0).
\end{align*} Since \(K\) is a constant, we get our desired result.
\end{proof}

\section{Lower Bound for All Streaming Models}
\label{sec:lb}
We will provide our lower bound in terms of \(\delta\), the average regret the algorithm incurs. We note that our lower bound is valid in the i.i.d.\  setting for discrete prediction, and consequently is a lower bound in all (adversarial, random order, i.i.d., and continuous costs) settings we consider in the paper.
Moreover, the lower bound holds even if the algorithm still has access to all \(\Omega(n)\) predictions of the experts when the outcome of the day is revealed to it (and loses access to the predictions only when it receives the next day's predictions). 

Our lower bound is achieved by reducing the problem of discrete prediction with expert advice to a \(n\)-fold version of the distributed detection problem we call \(\diffdist\). 

\begin{definition}[The \(\offset\shortdash\diffdist\) Problem]
	We have $T$ players, each of whom holds $n$ bits, indexed from $1$ to $n$.  We must distinguish between two cases, which we refer to as ``$V = 0$" and ``$V = 1$". Let \(\mu_0\) be a Bernoulli distribution with parameter \(\frac{1}{2}\), i.e., a fair coin, and let \(\mu_1\) be a Bernoulli distribution with parameter \(\frac{1}{2} + \offset\).
	\begin{itemize}
		\item (NO Case, ``$V = 0$") Every index for every player is drawn i.i.d. from a fair coin, i.e.,\ \(\mu_0\).
		\item (YES Case, ``$V = 1$") An index $L\in[n]$ is selected arbitrarily -- the $L$-th bit of each player is chosen i.i.d.\ from $\mu_1$. All other bits for every player are chosen i.i.d. from $\mu_0$.
	\end{itemize}
\end{definition}

Intuitively, each player in the \(\offset\shortdash\diffdist\) problem corresponds to a different day in the learning with experts problem. 
The $n$ bits held by each player correspond to the $n$ expert predictions for each day. 
Thus, in the NO case for the \(\offset\shortdash\diffdist\) problem, each expert is correct on half of the days in expectation (and with high probability) while in the YES case, there exists a single expert that is correct on a  $\frac{1}{2}+\offset$ fraction of the days in expectation (and with high probability). We note that we consider an algorithm that solves the \(\varepsilon\)-\diffdist problem with probability \(1 - p\), for some fixed constant \(p \in [0, 1]\), where this probability is over both the randomness of the input distribution, as well as the private randomness of the algorithm. 

\subsection{Communication Lower Bound of the \texorpdfstring{\(\varepsilon\)-\(\diffdist\)}{Epsilon-DiffDist} Problem}
Recall that we prove our lower bounds using the blackboard model, where each element of the stream is treated as a party (for us, each party will correspond to a day of predictions and corresponding outcome), and an algorithm for computing on the stream is seen as a multiparty communication protocol. Each of \(T\) parties has private randomness and communicates by posting a message to the blackboard, and we denote the transcript of all communication in a protocol by \(\transcript \in \left\{0, 1\right\}^*\). 
The communication cost of a protocol is the maximum bit length of the transcript, where the maximum is taken over all inputs and all coin tosses of the protocol. The communication complexity is the minimum communication cost of a correct protocol, where we will consider {\it distributional} correctness, meaning that a protocol is correct with failure probability $\gamma$ if it fails with probability at most $\gamma$, where the probability is taken over the joint distribution of the inputs and the protocol's private coins. We will take $\gamma$ to be a constant throughout, and will specify the input distributions we consider. We will also allow the protocol to have its own private coins, as we will need this when proving a direct sum theorem for information cost, described below. In the streaming model, the space complexity of an algorithm is the maximum amount of space in bits used by the algorithm. Note that any lower bound \(S\) on the randomized communication complexity in the blackboard model implies an \(S / T\) lower bound on the space complexity of a $1$-pass randomized streaming algorithm for solving the communication problem, since one player must communicate at least \(S/T\) bits. 

To prove a communication lower bound on the \(\varepsilon\)-\(\diffdist\) problem, we prove an analogue of the direct sum theorem \cite{bar2004information} that applies to the \(\varepsilon\)-distributed detection problem. The classic direct sum theorem from \cite{bar2004information} cannot be directly applied, since it is only applicable to decision problems, where the correct answer can be solely determined from the inputs. In our case the goal is to correctly infer a latent bit --- the correct answer is not a deterministic function of the input bits, but rather we are in a hypothesis testing scenario where we must infer the latent bit correctly with good probability under its respective posterior distribution. 

Hence, we will use a technique that is an analogue of the direct sum theorem in \cite{bar2004information}, but instead we directly show a lower bound on the mutual information in the case \(V = 0\). The mutual information $I(X ; Y)$ between two variables $X$ and $Y$ is equal to $H(X) - H(X | Y)$, or equivalently, $H(Y) - H(Y | X)$, where for a random variable $Z$, $H(Z)$ is the Shannon entropy of the distribution of $Z$. We refer the reader to \cite{bar2004information} for more background on information theory and the information complexity that we use. 

\begin{lemma}[Decomposable lemma]
Consider the distribution \(\mathbf{X} \sim \mu_0^n\) under the NO case of the \(\varepsilon\)-\(\diffdist\) problem. For any protocol \(\Pi\) that solves the \(\varepsilon\)-\(\diffdist\) problem with constant probability, the following mutual information inequality holds: \(I(\mathbf{X}; \Pi) \geq \sum\limits_{i = 1}^n I(\mathbf{X}_i; \Pi)\) under \(\mu_0^n\).
\label{lemma:DecomposableLemma}
\end{lemma}
\begin{proof}
First, we decompose mutual information into entropy: \hbox{\(I(\mathbf{X}; \Pi) = H(\mathbf{X}) - H(\mathbf{X} \mid \Pi) \)}. By independence of \(\mathbf{X}_i\) across \(i \in [n]\) under \(\mu_0^n\), we get that \hbox{\(H(\mathbf{X}) = \sum\limits_{i = 1}^n H(\mathbf{X}_i)\)}. On the other hand, by subadditivity of conditional entropy, we get that \hbox{\(H(\mathbf{X}, \Pi) \leq \sum\limits_{i = 1}^n H(\mathbf{X}_i, \Pi)\)}. Combining these two bounds, we arrive at our desired result.
\end{proof}

The decomposable lemma requires us to simply lower bound the mutual information of each coordinate with the transcript. Thus, we show that the mutual information can be lower bounded by using a lower bound on the mutual information of the \(\varepsilon\)-distributed detection problem.

\begin{lemma}[Reduction lemma]
For a protocol \(\Pi\) that solves the \(\varepsilon\)-\(\diffdist\) problem with constant probability at least \(1 - p\), where \(p \in [0, 0.5)\), the following inequality holds in the NO case of the \(\varepsilon\)-\(\diffdist\) problem: for every \(i \in [n]\), \(I(\mathbf{X}_i; \Pi | V = 0) \geq \Omega(\varepsilon^{-2})\).
\label{lemma:ReductionLemma}
\end{lemma}
\begin{proof}

First we note that the following setup is a reduction from the \(\varepsilon\)-distributed detection problem to the \(\varepsilon\)-\(\diffdist\) problem for any arbitrary choice of \(i \in [n]\):

Let \(\widetilde{\mathbf{X}}^{(t)}_i\) be i.i.d.\ samples from a fair coin for all \(j \neq i\) and \(t \in [T]\) (each player can sample their respective simulated inputs independently), and \(\widetilde{\mathbf{X}}_i = X^{(1)}, \dots, X^{(T)}\), where \(X^{(1)}, \dots, X^{(T)}\) is the input of the \(\varepsilon\)-distributed detection problem. Run the \(\varepsilon\)-\(\diffdist\) oracle on \(\widetilde{\mathbf{X}}\). Outputting \(V = 1\) if and only if the \(\varepsilon\)-\(\diffdist\) oracle outputs YES is a protocol for solving the \(\varepsilon\)-distributed detection problem with probability \(1 - p\) i.e.\ the input distribution to the \(\varepsilon\)-\(\diffdist\) oracle is in the YES case if and only if \(V = 1\).

Now that we have shown the reduction, consider the following two distributions:
\begin{enumerate}
    \item The distribution of \((\widetilde{\mathbf{X}}, \Pi(\widetilde{\mathbf{X}}))\) when \(V = 0\).
    
    \item The distribution of \((\mathbf{X}, \Pi(\mathbf{X}))\) for the \(\varepsilon\)-\(\diffdist\) problem in the NO case.
\end{enumerate}

These two distributions are equal i.e.\ both are the distribution where i.i.d.\ fair coins are drawn for each entry of \(\mathbf{X}\).  Thus, by the lower bound of \Cref{thm:EpsilonIC}, \(I(\mathbf{X}_i;\Pi(\mathbf{X})) = I(\widetilde{\mathbf{X}}_i; \Pi(\widetilde{\mathbf{X}})) \geq \Omega(\varepsilon^{-2})\). 

Since what we have shown above is true for every \(i \in [n]\), we get that \(I(\mathbf{X}_i;\Pi(\mathbf{X})) \geq \Omega(\varepsilon^{-2})\) for every \(i \in [n]\). Thus, we have achieved our desired result.
\end{proof}

Note that the \(\varepsilon\)-\(\diffdist\) problem is not solvable with \(o(\varepsilon^{-2})\) players (samples). We obtain the following randomized communication complexity lower bound for $\offset\shortdash\diffdist$, where recall
correctness is distributional, as described at the beginning of this section. 

\begin{lemma}
	The communication complexity of solving the \(\offset\shortdash\diffdist\) problem with a constant \(1 - p\) probability, for any \(p \in [0, 0.5)\), is $\Omega\left(\frac{n}{\offset^2}\right)$.
	\label{lemma:InfoOfDiffDist}
\end{lemma}
\begin{proof}
We have our desired lower bound by direct application of \Cref{lemma:DecomposableLemma,lemma:ReductionLemma} to lower bound the mutual information, and consequently communication, of a protocol that solves the \(\varepsilon\)-\(\diffdist\) problem with constant probability.
\end{proof}

\subsection{Reduction from \texorpdfstring{\(\diffdist\)}{DiffDist} to the Experts Problem}
We can now show a lower bound for the discrete prediction experts problem by reducing to it from the \(\offset\shortdash\diffdist\) problem.
Define an oracle algorithm, \(\oracle\), that achieves \(\delta\) average regret on the expert prediction problem with constant probability more than \(\frac{1}{2}\). Our goal is to show that we can solve the \(\offset\shortdash\diffdist\) problem with constant probability more than $\frac{1}{2}$ by using \(\oracle\). 
At a high level, we will treat each player \(i\)'s bit string as the predictions that a set of \(n\) ``experts'' made on day \(i\). 
To provide intuition for our reduction, we first describe a simpler reduction from \(\offset\shortdash\diffdist\) to the experts problem. The instance of the experts problem we construct has 1 as the correct answer on every day. We let each bit index correspond to an expert, and consequently, the predictions of the experts on a day \(i\) are the bits of player \(i\).

In the YES case of \(\offset\shortdash\diffdist\), there is an index, i.e.\ an expert, which is correct on approximately \(\frac{1}{2} + O(\avgregret)\) of the days. Thus, \(\oracle\) should also be correct \(\frac{1}{2} + O(\avgregret)\) of the time with high probability. On the other hand, in the NO case, the experts are all predicting uniformly randomly. Hence, the best expert does not do better than flipping a fair coin for its prediction. Our goal is to have \(\oracle\) have accuracy at least \(\frac{1}{2} + O(\delta)\) in the YES case and less than \(\frac{1}{2} + O(\delta)\) days in the NO case. 

However, while \(\oracle\) ensures an upper bound on \(\delta\), it makes no guarantees about the maximum accuracy \(\oracle\) can achieve. For example, a simple algorithm that simply predicts \(1\) on each day will achieve 100\% accuracy in both cases.  Thus, our simple reduction is insufficient since it cannot ``force'' \(\oracle\) to be sufficiently inaccurate in the NO case.

\paragraph{Masking in the Reduction.} To remedy this issue, we introduce a notion of ``masking'' i.e.\ obfuscating the the correct answer of each day in our construction so we can ensure an upper bound on the accuracy of \(\oracle\) when in the NO case of the \(\offset\shortdash\diffdist\) problem. In our actual reduction, formulated in \Cref{algorithm:DiffDistToExperts}, we compute a ``mask'' for each day by sampling a random bit from an independent fair coin. The mask XOR'ed with 1 will be the correct answer to the experts problem on that day. In addition, we also XOR the mask with the player's bits corresponding to that day to produce the experts predictions. 
This masking procedure ensures that all expert predictions and true outcomes are mutually independent in the NO case. That is, since the mask is drawn i.i.d.\ from a fair coin on each day, and the expert predictions are also drawn i.i.d.\ from a fair coin, the masked expert predictions remain distributed according to i.i.d.\ fair coins.  So, the true outcome on each day is distributed according to a fair coin that is completely independent of the expert predictions and past information provided to \(\oracle\).  Thus, \(\oracle\) can do nothing to increase (or decrease) its probability of success on each day from \(\frac{1}{2}\). 
On the other hand, in the YES case, there still remains an expert that is correct on \(\frac{1}{2} + \Theta(\avgregret)\) days, so \(\oracle\) will still get a \(\frac{1}{2} + \Omega(\avgregret)\) fraction of days correct.

\begin{algorithm}[!htb]
    \caption{The following algorithm is a reduction from \(\offset\shortdash\DiffDist\) to the experts problem where \(\oracle\) is an oracle algorithm that solves the experts problem with \(\delta\) regret and probability \(\frac{1}{2}\). Let \(c = \sqrt{2\ln(24)}\) and set \(\offset = \delta(c + 1)\), which we assume is less than $1/2$. Let \(\oplus\) be the XOR operation.}
    
    \KwIn{$\{\mathbf{X}^{(1)},\ldots,\mathbf{X}^{(T)}\}$, where $\mathbf{X}^{(t)} \in \{0, 1\}^n$ for each \(t \in [T]\).}
    \begin{algorithmic}
		\STATE{Let \(\texttt{state}_0\) be the initial state of $\oracle$.} 
		\FOR{each $t\in[T]$}
		    \STATE{Player \(t\) does the following:}
            \begin{ALC@g}
            
            \STATE{Sample $\texttt{mask}_t$ from an independent fair coin.}
		    
		    \STATE{$\texttt{maskedX}_t\gets \mathbf{X}^{(t)} \oplus (\texttt{mask}_t)^n$.}
		    
		    \STATE{Compute prediction and next state: $\texttt{prediction}_t, \texttt{state}_t \gets \oracle(\texttt{state}_{t - 1}, \texttt{maskedX}_t)$}
		    \STATE{\(\oracle\) is correct on day \(t\) iff $\texttt{prediction}_t \oplus \texttt{mask}_t  = 1$}
		    
		    \LineIf{\(t < T\)}{write \(\texttt{state}_{t}\) to the blackboard}
		    \end{ALC@g}
        \ENDFOR
        \STATE{Let \(S\) be the fraction sample days that are correct (each player communicates a bit indicating whether the algorithm was correct on each day).}
        \LineIf{\(S < \frac{1 + \avgregret c}{2}\)}{\(\mathbf{return}\) 0 \(\mathbf{else\ return}\) 1}
    \end{algorithmic}
    \label{algorithm:DiffDistToExperts}
\end{algorithm}

\thmlowerbound*
\begin{proof}
We show this proof for $\delta=\frac{1}{2+2\sqrt{2\ln(24)}}$ and $p=\frac{1}{4}$ by showing that \Cref{algorithm:DiffDistToExperts} is a valid reduction that solves \(\offset\shortdash\diffdist\) problem with at least \(\frac{3}{4}\) probability; the proof extends naturally to any constant $\delta$, $p<\frac{1}{2}$. 

Define \(c =\sqrt{2\ln(24)}\) for simplicity, and set \(\offset = \delta (c + 1)\). Since \(\frac{1}{2} + \offset\) is a valid probability, we require \(\frac{1}{2} + \offset \leq 1\). This implies an upper bound of \(\delta \leq \frac{1}{2(c + 1)}\). 

Let \(S\) be the accuracy of \(\oracle\), i.e., the proportion of days it is correct on. 
We will now show both cases of the \(\offset\shortdash\diffdist\) problem are solved with constant probability. In the NO case, \(S\) is simply the mean of \(T\) independent flips of a fair coin, since the correct answer on each day is independent of all other information. 

We can use Hoeffding's inequality to bound the probability that \(S\) exceeds the decision threshold, noting that, in this case, \(\mathbf{E}[S] = \frac{1}{2}\).
\begin{align*}
\prob{S \geq \frac{1 + \avgregret c}{2}} &\leq \exp \left( - \frac{c^2}{2} \right) \leq \frac{1}{3}.
\end{align*}

In the YES case, let \(E\) be the accuracy of the expert that is correct with probability \(\frac{1}{2} + \avgregret(c + 1)\). Thus, it is the mean of \(T\) i.i.d.\ Bernoullis with parameter \(\frac{1}{2} + \delta(c + 1)\). We know that the probability \(S \leq E - \delta\) is less than \(\frac{1}{4}\) by the guarantee on \(\oracle\). As a result, we can upper bound the error probability as follows:

\begin{align*}
\prob{{S} < \frac{1 + \avgregret c}{2}} &\leq \prob{E < \frac{1}{2} + \left(\frac{c}{2} + 1 \right) \delta} + \prob{S < E - \delta} \\
&\le 2 \exp \left(-\frac{c^2}{2}\right) + \frac{1}{4} \quad < \frac{1}{3}.
\end{align*} 
The first inequality is by conditional probability being larger than joint probability and the union bound.  
The final line is by Hoeffding's inequality, \(\delta^2 T \ge 1\), and our choice of \(c\).

Thus, we have shown for both cases, there is at least a \(\frac{2}{3}\) probability of solving the \(\offset\shortdash\diffdist\) problem, and hence \Cref{algorithm:DiffDistToExperts} is a valid reduction. As a result, the total communication cost of \Cref{algorithm:DiffDistToExperts} is \(\Omega( \frac{n}{\avgregret^2})\) by \Cref{lemma:InfoOfDiffDist}.

Let \(C\) be the total communication cost of writing each \(\texttt{state}_t\) to the blackboard in \Cref{algorithm:DiffDistToExperts}. That is, the total communication of \Cref{algorithm:DiffDistToExperts} is \(C + T\), \(C\) for the \(\texttt{state}_t\)s and \(T\) for the correct outcomes.  We can assume that \(T \in o\left(\frac{n}{\delta^2}\right)\), otherwise our bound is trivial.  So this means that \(C \in \Omega(\frac{n}{\delta^2})\), and we have shown our desired lower bound on the communication of \(\oracle\). A streaming lower bound of \(\Omega(\frac{n}{\delta^2T})\) follows directly from this communication lower bound, which completes our proof for $\delta=\frac{1}{2+2\sqrt{2\ln(24)}}$ and $p=\frac{1}{4}$.
The proof extends naturally to all $\delta,p<\frac{1}{2}$.   
\end{proof}

\section{Prediction with Experts in the Standard Streaming Model}
\label{sec:maj}
As a warm-up to our near-optimal algorithm for online learning with experts in the random-order model, we show an algorithm that can handle arbitrary-order streams.  
To build intuition, we first consider the simpler discrete prediction with experts problem, and then generalize our algorithm for general costs in \([0, 1]\).  
Recall that in the discrete prediction problem, the cost of each decision is either $0$ or $1$. 
We first propose a space constrained version of the simplest version of the majority elimination algorithm where experts that are incorrect for a ``significant'' fraction of days are eliminated. 
The algorithm uses the desired $\tilde{O}\left(\frac{n}{T}\right)$ space complexity for correctness on a constant fraction of days, even when the best expert makes as many as $O(T/\log^2 n)$ mistakes.
When errors on only a $\delta$ fraction of days are permitted for some subconstant $\delta$, the algorithm uses $\tilde{O}\left(\frac{n}{T\delta}\right)$ space but demands that the number of mistakes made by the best expert is at most $O\left(\frac{\delta T}{\log^2 n}\right)$. 

\subsection{Discrete Prediction in the Standard Streaming Model}

Our algorithm for this upper bound is to run the majority elimination algorithm in small chunks at a time.  The typical majority elimination algorithm (not in the space constrained setting) maintains a voting pool of experts, that starts by including all of the experts.  On each day, the algorithm predicts the majority vote of the experts in the pool.  When the outcome is revealed, the algorithm removes any expert who made an incorrect prediction from the pool.  If the pool is empty, the algorithm resets by adding all of the experts back into the pool.  In this formulation, the algorithm makes at most \(\log n\) times as many mistakes as the best expert.  We modify this algorithm to work in the space constrained setting, and impose a laxer requirement on the experts in our pool to achieve a better bound.

The algorithm proceeds in rounds. 
At the beginning of each round, the algorithm initializes a pool, $P$, of $k = \frac{16n\log^2 n}{T\delta}$ different experts and the variable $u$ to mark the time that the round begins. 
On each day, the algorithm outputs the majority vote of the experts in the pool. 
Any expert that is inaccurate for at least $\frac{\delta}{8 \log n}$ fraction of the days since time $u$ is deleted. 
Once the pool $P$ is completely depleted of experts, the current round ends and the next round begins. 
A complete description is given in Algorithm~\ref{alg:SublinearAlgorithm}. 

\begin{algorithm}[!htb]
    \caption{An expert algorithm that maintains a pool of experts occupying \(\widetilde{O}(n / T \delta)\) space, and eliminates an expert if its accuracy drops below $1 - \frac{\delta}{8\log n}$. This is algorithm is a streaming analogue of the typical majority elimination for experts.}
		\label{alg:SublinearAlgorithm}
    \KwIn{Number $n$ of experts, number $T$ of rounds, fraction $1-\delta$ of mistakes}
    \begin{algorithmic}[1]
                \STATE{$k\gets\frac{16n\log^2 n}{T\delta}$, $u\gets 1$}
				\STATE{Let \(P\) be a random set of $k$ unique indices of $[n]$.}
				\FOR{each time $t\in[T]$}
				\STATE{Output the majority vote of the experts in \(P\).}
				\STATE{Discard any experts in $P$ with lower than \(1 - \frac{\delta}{8\log n}\) accuracy since time $u$.}
				\IF{$P=\emptyset$}
				\STATE{Let $P$ be a random set of $k$ unique indices of $[n]$.}
				\STATE{$x\gets 0$, $u\gets t$}
        \ENDIF
				\ENDFOR
    \end{algorithmic}
\end{algorithm}

\noindent
We first bound the number of mistakes made by the algorithm in a particular round. 
\begin{lemma}
\label{lem:mistakes}
Fix a \(\delta > \frac{16 \log^2 n}{T}\), and 
suppose a pool, $P$, of size $k= \frac{16n\log^2 n}{T\delta}$ 
is initiated by Algorithm~\ref{alg:SublinearAlgorithm} at time $t_0$ and $P\neq\emptyset$ before some later time $t$. 
Then the number of mistakes by the algorithm between times $t_0$ and $t$ is at most $\frac{(t-t_0)\delta}{2}+4\log n$.
\end{lemma}
\begin{proof}
Let $u_1,\ldots,u_y$ be a sequence of times defined so that $u_i$ is the first time at which at most $\frac{k}{2^i}$ experts remain in the pool. 
Note that \(y \leq \lceil \log k \rceil\) by definition.
Let $n_i$ be the total number of mistakes the algorithm has made by time $u_i$. 
We note that $n_{i+1}-n_i$ mistakes are made by the algorithm between times $u_i$ and $u_{i+1}$, and each mistake requires at least $\frac{k}{2^{i+2}}$ mistakes across all experts, since there at least $\frac{k}{2^{i + 1}}$ experts before time \(u_{i + 1}\), and at least half of them must be wrong for a mistake to be made. Consequently, the total number of mistakes made by the experts between times $u_i$ and $u_{i+1}$ is at least $(n_{i+1}-n_i)\cdot\frac{k}{2^{i+2}}$. 
We can also see that at most $\left\lceil\frac{(u_{i+1}-t_0)\delta}{8\log n}\right\rceil$ mistakes can be made by each of the $\frac{k}{2^i}$ experts that are not deleted by time $u_{i}$, so the total number of mistakes made by the experts between times $u_i$ and $u_{i+1}$ is at most 
\[\left\lceil\frac{(u_{i+1}-t_0)\delta}{8\log n}\right\rceil\cdot\frac{k}{2^i}\le\frac{(u_{i+1}-t_0)k\delta}{8\cdot 2^i\log n}+\frac{k}{2^i}.\] 
Hence we have $(n_{i+1}-n_i)\cdot\frac{k}{2^{i+2}}\le\frac{(u_{i+1}-t_0)k\delta}{8\cdot 2^i\log n}+\frac{k}{2^i}$ so that
\[(n_{i+1}-n_i)\le\frac{(u_{i+1}-t_0)\delta}{2\cdot\log n}+4.\]
Therefore,
\[\sum_{i=1}^{y-1} \left(n_{i+1}-n_i\right)\le\sum_{i=1}^{y-1}\left(\frac{(u_{i+1}-t_0)\delta}{2\cdot\log n}+4\right)\le\frac{(t-t_0)\delta}{2}+4\log k.\]
Thus, we have shown our desired result.
\end{proof}
We now give the full guarantees for our algorithm. 
\begin{theorem}
Fix a \(\delta > \frac{16 \log^2 n }{T}\), and suppose the best expert makes at most $M\le\frac{\delta T}{128\log^2 n}$ mistakes.  
Then Algorithm~\ref{alg:SublinearAlgorithm} for the discrete prediction with experts problem uses $\widetilde{O}\left(\frac{n}{\delta T}\right)$ space and achieves regret at most $\delta$, with probability at least $4/5$. 
\label{thm:discrete-prediction-alg}
\end{theorem}
\begin{proof}
Algorithm~\ref{alg:SublinearAlgorithm} only makes more than $\delta T$ total mistakes if it completes at least $\frac{\delta T}{8\log n}$ rounds. This is because, by \Cref{lem:mistakes} and the lower bound on \(\delta\), the bound on mistakes after $\frac{\delta T}{8\log n}$ rounds is
\[\frac{T\delta}{2}+\frac{T\delta}{8\log n}\cdot4\log n\le T\delta.\] 

Thus our goal is to show that the probability that \Cref{alg:SublinearAlgorithm} completes at least $\frac{\delta T}{8\log n}$ rounds is low. We will accomplish this by arguing that probability of the best expert being either (1) sampled and deleted or (2) never sampled in the first $\frac{\delta T}{8\log n}$ rounds is low. 
We can show that the first case can only happen with small probability through a standard calculation. 
We would like to similarly say that conditioned on a large number of rounds, the probability that the algorithm sampled and subsequently deleted the best expert is low. 
Unfortunately, the conditional event that the algorithm completes a large number of rounds can significantly alter the distribution of events, resulting in a more involved analysis. 

We thus use a decoupling argument for the distribution of times to instead analyze the distribution for the total number of rounds. 
Let $d_0,d_1,d_2,\ldots$ be the random variables representing the times on which the pool of $k$ experts is sampled, so that $d_i$ is the random variable for the time on which the pool of experts is empty for the $i$-th time and thus resampled. 
We construct a sequence $t_0,t_1,t_2,\ldots$ of times so that the distribution of $t_0,t_1,t_2,\ldots$ will match the distribution of $d_0,d_1,d_2,\ldots$. 
Let $t_0=0$ and for each $i>0$, let $t_i$ be drawn uniformly from the distribution of possible times at which a new pool of $k$ experts drawn at time $t_{i-1}$ is completely removed from the pool. 
The sequence $\{t_i\}_i$ is terminated if the experts drawn at time $t_{i-1}$ are not all removed from the pool by time $T$. 
Note that the process in which the sequence $t_0,t_1,\ldots$ is generated matches the process in which the sequence $d_0,d_1,\ldots$ is determined, so that their distributions are identical so we can instead work with the random sequence $t_0,t_1,\ldots$.

Let $\bad$ be the set of times on which the best expert would be eliminated by the algorithm if it were added to pool of experts on a time $t\in \bad$. 
Because the best expert makes at most $M$ mistakes and the algorithm deletes experts that have made mistakes on at least $\frac{\delta}{8\log n}$ fraction of the times since they have been in the pool, then it follows that $|\bad|\le\frac{8M\log n}{\delta}$. 

Let $R$ be the random variable that corresponds to the total number of rounds in the algorithm. 
Let $B$ be the random variable that corresponds to the total number of rounds in the algorithm that started on days $t$ such that $t\in \bad$. 
Let $G$ be the random variable that corresponds to the total number of rounds in the algorithm that started on days $t$ such that $t\notin \bad$. 
Observe that $B\le|\bad|\le\frac{8M\log n}{\delta}$ and $R\le B+G\le\frac{8M\log n}{\delta}+G$. 
Since $M\le\frac{\delta^2 T}{128\log^2 n}$, then 
\[R\le\frac{\delta T}{16\log n}+G.\]
Thus it remains to analyze the distribution of $G$. 

Observe that if the best expert is sampled on a time $t$ with $t\notin\bad$, then the best expert will not be deleted and thus there will be no subsequent round. 
Hence if $G=j$ for some $j\ge 1$, then the best expert must not have been sampled into the first $j-1$ pools that were sampled on times $t_1,\ldots,t_{j-1}\notin \bad$. 
It follows that $\PPr{G \ge j}\le(1-k/n)^{j - 1}$ for each integer $j\ge 1$. Since $k=\frac{16n\log^2 n}{T\delta}$, we have
\begin{align*}
 \PPr{G \ge \frac{\delta T}{16 \log n}} &\le \left(1 - \frac{16 n \log^2 n}{T \delta} \cdot \frac{1}{n} \right)^{\frac{\delta T}{16 \log n}} \\
 &\le \left(1 - \frac{16 \log^2 n}{T \delta}\right)^{\frac{T \delta}{16 \log^2 n} \cdot \log n} \\
 &\le e^{-\log n} 
 \qquad \qquad \le \frac{1}{n} 
\end{align*}

Therefore, $\PPr{R\ge\frac{\delta T}{16\log n}+\frac{\delta T}{16\log n}}\le\frac{1}{n}$. 
Hence,
\[\PPr{R<\frac{\delta T}{8\log n}}\ge 1-\frac{1}{n},\]
which is \(\ge \frac{4}{5}\) for sufficiently large \(n\). 
Conditioned on the event that $R<\frac{\delta T}{8\log n}$, then by Lemma~\ref{lem:mistakes}, the total the number of mistakes by the algorithm is at most 
\[\frac{T\delta}{2}+4R\log n\le\frac{T\delta}{2}+\frac{T\delta}{8\log n}\cdot4\log n\le T\delta.\] 
The average regret of the algorithm is defined as the difference between the error rate of the algorithm and the optimal error rate (the error rate of the best expert). 
This is upper bounded by the error rate of the algorithm, which is at most $T\delta/T = \delta$.
\end{proof}

\subsection{Online Prediction for General \texorpdfstring{$[0, 1]$}{[0,1]} Costs in the Standard Streaming Model}

Now, we show that with a simple change, we can modify this algorithm to work for general \([0, 1]\) costs. The trouble with this algorithm for \([0, 1]\) costs is that it is no longer clear what it means to take the ``majority prediction" on each day.  However, there are other sequential prediction algorithms that we could use, that work well in the case of general costs. Thus, for \([0, 1]\) costs, we can implement Algorithm~\ref{alg:SublinearAlgorithmGeneralCosts}.

\begin{algorithm}[!htb]
    \caption{An expert algorithm that works for \([0, 1]\) costs in the standard streaming model.}
		\label{alg:SublinearAlgorithmGeneralCosts}
    \KwIn{Number $n$ of experts, number $T$ of rounds, fraction $1-\delta$ of mistakes, (black-box sequential prediction algorithm of choice)}
    \begin{algorithmic}[1]
        \STATE{$k\gets\frac{16 \beta n\ln n \ln T}{T\delta} \in \widetilde{O}(\frac{n}{T \delta})$, $u\gets 1$}
		\STATE{Let \(P\) be a random set of $k$ unique indices of $[n]$.}
		\STATE{Initialize a sequential prediction algorithm (Definition \ref{def:seq-pred-alg}) for the experts in \(P\) with \(\eps = \frac{1}{2}\)}
		\FOR{each time $t\in[T]$}
    		\STATE{Output prediction according to sequential prediction algorithm running on \(P\).}
    		\IF{every expert in \(P\) has an error rate higher than \(\frac{\delta}{8}\) since time \(u\)}
        		\STATE{Let $P$ be a random set of $k$ unique indices of $[n]$.}
        		\STATE{$x\gets 0$, $u\gets t$}
        		\STATE{Re-initialize sequential prediction algorithm for experts in \(P\) with \(\eps = \frac{1}{2}\)}
            \ENDIF
		\ENDFOR
    \end{algorithmic}
\end{algorithm}

Our analysis follows the same structure as the one for discrete costs.  We can use the guarantee of the sequential prediction algorithm (Definition \ref{def:seq-pred-alg}) in place of Lemma \ref{lem:mistakes}.  The bound on the number of ``bad days" that is used in Theorem~\ref{thm:discrete-prediction-alg} is now somewhat more involved and is presented as its own lemma.  

\begin{lemma}
Let \(\bad\) be the set of times on which the best expert would be eliminated by the algorithm if it were added to a pool of experts on a time \(t \in \bad\), and suppose the best expert incurs total cost at most \(M\).  Then, 
\[|\bad|\le \frac{8 M}{\delta}.\]
\label{lem:bad-days}
\end{lemma}
\begin{proof}
We show this via amortized analysis. 
For convenience, denote the threshold \(L = \frac{\delta}{8}\) and let $m_i$ be the cost of the best expert on day $i$ for each $i\in[T]$. 
Let $\{t_1,\ldots,t_b\}$ be the days in $\bad$, so that we would like to show that $b=|\bad|\le\frac{8 M}{\delta}$. 
For each $t_j\in\bad$ with $j\in[b]$, i.e., the $j$-th bad day, let $e_j$ be the day on which the best expert would be deleted if it were sampled on day $t_j$, so that $\sum_{i=t_j}^{e_j} m_i\ge (e_j-t_j+1)L$.

We define a subsequence $t_{a_1},\ldots,t_{a_{b'}}$ of $t_1,\ldots,t_b$ as follows. 
Let $a_1=1$ and for each $k>1$, let $a_k=\min_{i\in[b]}\{i\,:\,t_i>e_{a_{k-1}}\}$. 
In other words, $t_{a_k}$ is the first bad day after the deletion of the $(k-1)$-th term in the subsequence. 
Thus the sequence $t_{a_1},\ldots,t_{a_{b'}}$ is a subsequence of $t_1,\ldots,t_b$ and the intervals $[t_{a_1},e_{a_1}],\ldots,[t_{a_{b'}},e_{a_{b'}}]$ are disjoint, so we have that \( \sum_{i = 1}^T m_i \ge \sum_{j = 1}^{b'} \sum_{i = t_{a_j}}^{e_{a_j}} m_i\).
We also know that each of the \(b\) bad days is in one of these intervals, so $\sum_{j=1}^{b'} e_{a_j}-t_{a_j}+1\ge b$.

Therefore, we have 
\[ M = \sum_{i = 1}^T m_i \ge \sum_{j = 1}^{b'} \sum_{i = t_{a_j}}^{e_{a_j}} m_i \ge \sum_{j=1}^{b'} (e_{a_j}-t_{a_j}+1)L\ge bL.\]

It follows that $b\le\frac{M}{L}$, where $b=|\bad|$ is the number of bad days and $L=\frac{\delta}{8}$. 
Thus,
\[|\bad|\le \frac{8 M}{\delta}.\]
\end{proof}

\noindent This allows us to give the full guarantees for the algorithm for general costs.

\begin{theorem}
Let \(\delta > \frac{16 \beta \ln^2 n}{T}\), and suppose the best expert makes at most $M\le\frac{\delta^2 T}{128 \beta \ln n}$ mistakes.  
Then Algorithm~\ref{alg:SublinearAlgorithmGeneralCosts} for the online learning with experts problem uses $\widetilde{O}\left(\frac{n}{\delta T}\right)$ space and achieves regret at most $\delta$ in expectation.

Recall that \(\beta\) is a fixed constant that depends on the black-box sequential prediction algorithm that is used (Definition \ref{def:seq-pred-alg}).
\label{thm:general-prediction-alg}
\end{theorem}

\begin{proof}
This proof follows the same structure as the proof of Theorem~\ref{thm:discrete-prediction-alg}. 
First, we note that for our algorithm to be well defined, we require that \(k \le n\), which means that we require
\[ \frac{16 \beta n \ln^2 n}{T \delta} \le n \Longrightarrow \frac{16 \beta \ln^2 n}{T} < \delta.\]

Again, let $R$ correspond to the total number of rounds in the algorithm, let \(B\) be the number of rounds that start on bad days, and \(G\) be the number of rounds that start on good days, so \(R = B + G\).  By Lemma \ref{lem:bad-days}, we know that \(B \le |\textsc{Bad}| \le \frac{8M}{\delta} \le \frac{\delta T}{16 \beta \ln n} \).  

We also know that if \(G = j\), for some \(j \ge 1\) then the best expert was not sampled on the first \(j - 1\) rounds that started on good days.  So \(\PPr{G \ge j} \le \left(1 - k/n\right)^{j - 1}\).  Since \(k = \frac{16 \beta n \ln n \ln T}{T \delta}\), this means 
\begin{align*}
    \PPr{G \ge \frac{\delta T}{16 \beta \ln n}} &\le \left(1 - \frac{16 \beta n \ln n \ln T}{T \delta} \cdot \frac{1}{n} \right)^{\frac{\delta T}{16 \beta \ln n}} \\
    &\le \left(1 - \frac{16 \beta \ln n \ln T }{T \delta} \right)^{\frac{\delta T}{16 \beta \ln n \ln T} \cdot \ln T} \\
    &\le e^{-\ln T} \qquad \le \frac{1}{T}.
\end{align*}

So this gives us
\[\PPr{R < \frac{\delta T}{8 \beta \ln n}} \ge 1 - \frac{1}{T}.\]

Now, conditioning on the event that \(R < \frac{\delta T}{8 \beta \ln n}\), we can bound the expected cost of our algorithm.  Consider a round \(\mathbf{r}\) that goes from time \(\mathbf{r}_{\text{start}}\) to time \(\mathbf{r}_{\text{end}}\).  We know by the guarantee for the sequential prediction algorithm (Definition \ref{def:seq-pred-alg}),
\begin{align*}
    \textbf{E}\Bigg[ \text{cost of seq.\ pred.\ alg.\ on } \mathbf{r} \Bigg] &\le \left(1 + \frac{1}{2}\right) \Bigg[ \text{cost of best expert in pool over } \mathbf{r} \Bigg] + 2 \beta \ln k.
\end{align*}
Our condition for running this round was that until the last day, when we decided to resample, there was at least one expert who had at most \(\frac{\delta}{8}\) average cost.  The cost on that last day is trivially upper bounded by 1.  So we have 
\begin{align*}
    \textbf{E}\Bigg[ \text{cost of seq. pred. alg. on } \mathbf{r} \Bigg] &\le \frac{3}{2} \Bigg[ \frac{\delta}{8} \cdot (\mathbf{r}_{\text{end}} - \mathbf{r}_{\text{start}}) \Bigg] + 1 + 2 \beta \ln k
\end{align*}
Summing over all rounds, this gives us
\begin{align*}
    \textbf{E}\Bigg[ \text{cost of algorithm} \Bigg] &\le \sum_{\text{rounds } \mathbf{r}}  \left[ \frac{3}{16} \cdot\delta \cdot (\mathbf{r}_{\text{end}} - \mathbf{r}_{\text{start}})  + 1 + 2 \beta \ln k \right]\\
    &\le \frac{3}{16} \cdot \delta T + R \cdot (1 + 2 \beta \ln n) \\
    &\le \frac{3}{16} \cdot \delta T + \frac{\delta T}{8 \beta \ln n} \cdot (1 + 2 \beta \ln n)\\
    &\le \frac{\delta}{2} \cdot T.
\end{align*}
So with probability at least \(1 - \frac{1}{T}\), our algorithm has average cost at most \(\delta/2\) in expectation, and with probability at most \(\frac{1}{T}\), our algorithm can have cost up to \(T\).  So this algorithm achieves at most \(\delta\) cost, and therefore at most \(\delta\) regret in expectation.    
\end{proof}

\section{General Costs in the Random-Order Streaming Model}
\label{sec:mw}
In this section, we consider the online learning with experts problem, in which the cost of the decision of an expert on each day can range from $[0,\rho]$, where $\rho>0$ is the width of the problem. 
Without loss of generality, we assume $\rho=1$ throughout this section and instead incur a multiplicative factor in the regret in the guarantees of our algorithms, i.e., our algorithms will have regret $\rho\delta$ rather than $\delta$. 
Whereas the previous algorithm provided guarantees on arbitrary streams, our main algorithm in this section will focus on the random-order streaming model. 

The main result in the previous section, for arbitrary-order streams, relied on an assumption that the best expert incurred sub-constant average cost.  This allowed us to conclude that there were not too many ``bad" days in the stream, where a ``bad" day is one where if we start a round with the best expert on that day, the best expert will appear to do badly, causing the round to end.

For random-order streams, this is no longer a problem, because the best expert will effectively do uniformly well across the entire stream.  This means that any day on which we sample the best expert will likely be a ``good" day.  This allows us to remove the condition on the best expert.  

We will first show an algorithm that achieves \(\delta\) regret, if it knows \(M\), the number of mistakes  made by the best expert.  Then we show how to modify the algorithm to include a searching phase which allows us to estimate \(M\) and so the algorithm does not need to know it in advance.
The algorithm for arbitrary-order streams did not need to know \(M\) in advance because we assumed an upper bound on $M$.  
However, for this algorithm, we remove this upper bound assumption on \(M\), though we will need to look at a prefix of days to estimate \(M\) for use in our algorithm.

\begin{algorithm}[!htb]
    \caption{An expert algorithm that maintains a pool of experts occupying \(\widetilde{O}(n / (\delta^2 T))\) space, and resamples the pool if its expected cost is too high.}
		\label{alg:SublinearAlgorithm:MW}
    \KwIn{Number $n$ of experts, number $T$ of rounds, regret $\delta$, number $M$ of mistakes of the best expert. We later show how to instead estimate $M$.}
    \begin{algorithmic}[1]
    \STATE{$u\gets 1$, $k\gets O\left(\frac{n\log^2 n}{\delta^2 T}\right)$}
	\STATE{Let \(P\) be a random set of $k$ unique indices of $[n]$.}
	\FOR{each time $t\in[T]$}
        \STATE{Run a sequential prediction algorithm (Definition \ref{def:seq-pred-alg}) with \(\eps = \delta/2\) for the experts in \(P\).}
        \IF{the cost of every expert in the pool exceeds\\ $\frac{M}{T}(t-u)+4\sqrt{(t-u)\log T}$ since time $u$}
        	\STATE{Let $P$ be a random set of $k$ unique indices of $[n]$.}
        	\STATE{$u\gets t$}
        \ENDIF
	\ENDFOR
    \end{algorithmic}
\end{algorithm}

We first note that the best expert in the random-order model cannot incur high cost. The following is Hoeffding's bound. 
\begin{lemma}
\label{lem:best:upper:random}
Let $X_1,\ldots,X_t$ be independent random variables such that $X_i\in[0,1]$ with $\mathbb{E}[X_i]=\alpha$ for all $i\in[t]$ and let $X=\sum_{i=1}^t X_i$. 
Then for any $T>1$, $\PPr{|X-\alpha t|\ge4\sqrt{t\log T}}\le\frac{1}{T^2}$.
\end{lemma}

We can apply Lemma~\ref{lem:best:upper:random} in conjunction with the distributional properties of random-order streams and a union bound to show that with high probability, a pool with the best expert will be retained.

\begin{corollary}
\label{cor:no:resample}
In the random-order model, with probability at least $1-\frac{1}{T}$, Algorithm~\ref{alg:SublinearAlgorithm:MW} will not resample a pool including the best expert.
\end{corollary}

\begin{proof}
Suppose the best expert is sampled by a pool. 
Let $\alpha=\frac{M}{T}$, where $M$ is the total cost of decisions made by the best expert. 
Observe that Lemma~\ref{lem:best:upper:random} considers a setting where the best expert has i.i.d. cost each day with expectation $\alpha$. 
In the random-order model, the cost of each day will follow a multivariate hypergeometric distribution whose expectation on each day is $\alpha$. 
Hence, the same bound as Lemma~\ref{lem:best:upper:random} holds, so that by a union bound, the best expert will incur cost at most $\alpha t+4\sqrt{t\log n}+\frac{4\log n}{\alpha}$ across all times $t$ after the best expert is sampled by the pool, with probability at least $1-\frac{1}{T}$. 
Thus the pool will not be resampled with probability at least $1-\frac{1}{T}$. 
\end{proof}

We now analyze Algorithm~\ref{alg:SublinearAlgorithm:MW}, which assumes that the cost $M$ of the best expert is given as input to the algorithm. 
We remove this assumption afterwards. 
\begin{theorem}
\label{thm:regret:known:mistakes}
For any $\delta>\sqrt{\frac{16\log^2 n}{T}}$, there exists an algorithm that takes as input a number $M$, which is the cost of the best expert, and achieves regret at most $\delta$ in expectation on random-order streams.  
The algorithm uses $O\left(\frac{n\log^2 n}{\delta^2T}\right)$ space. \end{theorem}
\begin{proof}
Consider Algorithm~\ref{alg:SublinearAlgorithm:MW} and suppose by way of contradiction, that its expected cost is at least $M+\delta T$. 
Let $\alpha=\frac{M}{T}$. 
Suppose the $j$-th pool of experts was run for time $t_j$. 
Then the best expert in the pool has cost at most $\alpha t_j+4\sqrt{t_j\log n}+1$. 
Then by Definition \ref{def:seq-pred-alg} for $\eps=\frac{\delta}{2}$, the expected cost of running the sequential prediction algorithm on the $j$-th pool of experts is at most 
\[\left(1+\frac{\delta}{2}\right)\left(\alpha t_j+4\sqrt{t_j\log n}+1\right)+\frac{2 \beta \ln n}{\delta},\]
for some fixed constant \(\beta\).
Thus, if there are $r$ total rounds over time $t$, the expected cost of the algorithm by linearity of expectation is at most
\begin{align*}
\left(1+\frac{\delta}{2}\right)\left(\alpha t+4\sqrt{rt\log n}+r\right)+\frac{2 \beta r\ln n}{\delta}\\
=\alpha t+O\left(\delta\alpha t+\sqrt{rt\log n}+\frac{r\log n}{\delta}\right).
\end{align*}
Hence, if the expected cost is at least $M+\frac{3\delta T}{4}$ and $\delta>\sqrt{\frac{16\log^2 n}{T}}$, then the algorithm must have used $r=\Omega\left(\frac{\delta^2 T}{\log n}\right)$ rounds. 

We now analyze the probability distribution for the number of rounds that the algorithm uses. 
By Corollary~\ref{cor:no:resample}, any pool that includes the best expert will not be resampled in the random-order model, with probability at least $1-\frac{1}{T}$. 
Thus, conditioning on the event that the first pool sampled that includes the best expert is not resampled, then if the algorithm uses $j$ total rounds, the first $j-1$ rounds must have not sampled the best expert. 
Therefore, if $Z$ is a random variable that represents the number of rounds, we have $\PPr{Z\ge j}\le\left(1-\frac{k}{n}\right)^{j-1}$. 
Since $k=O\left(\frac{n\log^2 n}{\delta^2 T }\right)$ with a sufficiently large constant in the big-Oh, then $\PPr{Z\ge r}\le\frac{1}{\poly(T)}$ for $r=\Omega\left(\frac{\delta^2 T}{\log n}\right)$.  
Hence, we have that with probability at least $1-\frac{2}{T}$, the expected cost of the algorithm is at most $M+\frac{3\delta T}{4}$. 
Otherwise, the cost of the algorithm is at most $T$. 
Thus, the overall expected cost of the algorithm is at most $M+\delta T$. 
\end{proof}

\paragraph{Unknown cost of the best expert in the random-order model.} 
We remark that Algorithm~\ref{alg:SublinearAlgorithm:MW} assumes the cost $M$ incurred by the best expert is known. 
We now describe how this assumption can be easily removed in the random-order model. 
Note that since the overall expected cost of Algorithm~\ref{alg:SublinearAlgorithm:MW} is at most $M+\delta T$, then even if we use a $(1+O(\delta))$-approximation of $M$ as input to the algorithm, then the overall expected cost is $(1+O(\delta))M+\delta T=M+O(\delta T)$, which can be then adjusted to $M+\delta T$ by a rescaling of $\delta$. 
Thus, it suffices to find a $(1+O(\delta))$-approximation to $M$. 

Let $\gamma$ be an estimate for the average cost $\frac{M}{T}$, and we initialize $\gamma$ to $\frac{1}{2}$. 
Note that a $(1+O(\delta))$-approximation to $\gamma$ corresponds to a $(1+O(\delta))$-approximation to $M$. 
We obtain a $(1+O(\delta))$-approximation to $\gamma$ through a binary search. 
We proceed through $\ell:=2\log\frac{1}{\delta}$ epochs so that in each epoch $j\in[\ell]$, $\gamma$ is a $\left(1+\frac{1}{2^j}\right)$-approximation to $\frac{M}{T}$. 
Each epoch $j\in[\ell]$ has length $\frac{\delta T}{2\log\frac{1}{\delta}}$. 
We run Algorithm~\ref{alg:SublinearAlgorithm:MW} on this epoch with input $\gamma\cdot\frac{\delta T}{2\log\frac{1}{\delta}}$ as the estimate for the cost and $\frac{1}{100}$ as the target regret. 
We can also track the average cost $\beta_j$ of the best expert in each epoch $j$. 
If $\gamma>(1+\delta)\beta_j$, then we update $\gamma\gets\gamma-\frac{1}{2^{j+1}}$. 
Similarly if $\gamma<(1-\delta)\beta_j$, then we update $\gamma\gets\gamma+\frac{1}{2^{j+1}}$. 
After the $\ell$ epochs, we will fix $\gamma$ as the estimated average cost for the remainder of the stream and run Algorithm~\ref{alg:SublinearAlgorithm:MW}. 

\begin{theorem}
\label{thm:general:regret}
For any $\delta>\sqrt{\frac{16\log^2 n}{T}}$, there exists an algorithm that achieves regret at most $\delta$ in expectation on random-order streams. 
The algorithm uses $O\left(\frac{n}{\delta^2 T}\log^2 n\right)$ space. 
\end{theorem}
\begin{proof}
It suffices to (1) show that $\gamma$ converges to a $(1+\delta)$-approximation of the true average cost $\frac{M}{T}$ by the best expert and (2) analyze the regret induced by the procedure until $\gamma$ converges. 
The expected regret of the algorithm afterward is upper bounded by Theorem~\ref{thm:regret:known:mistakes}. 

To show that $\gamma$ converges to a $(1+\delta)$-approximation of the true average cost $\frac{M}{T}$ by the best expert, we consider casework on $\gamma$. 
Suppose $\gamma>(1+\delta)\cdot\frac{M}{T}$. 
Then by Lemma~\ref{lem:best:upper:random}, no experts sampled by the pool will achieve average cost $\gamma$ in the random-order model. 
Thus $\gamma$ will be decreased accordingly. 
On the other hand, if $\gamma<(1-\delta)\cdot\frac{M}{T}$, then again by Lemma~\ref{lem:best:upper:random}, the best expert will be sampled by the pool and have average cost at least $(1-\delta)\cdot\frac{M}{T}$ in the random-order model and then $\gamma$ will be increased accordingly. 
Hence, with probability at least \(1 - \frac{1}{T}\), it holds that $\gamma$ converges to a $(1+\delta)$-approximation of the true average cost $\frac{M}{T}$ of the best expert. 

On the other hand, since each epoch $j\in[\ell]$ only has length $\frac{\delta T}{2\log\frac{1}{\delta}}$ and $\ell= 2\log\frac{1}{\delta}$, then the total cost that can be incurred across the $\ell$ epochs is only $\delta T$. 
Hence the regret can only be increased by an additive $\delta$ due to not knowing the average cost of the best expert. 
Finally, to analyze the space complexity, recall that each epoch has length $\frac{\delta T}{2\log\frac{1}{\delta}}$ and $\frac{1}{100}$ is the target regret. 
With high probability, the best expert makes at least $\frac{\delta M}{200\log\frac{1}{\delta}}$ mistakes in each epoch.  
Thus by Theorem~\ref{thm:regret:known:mistakes}, it suffices to use $O\left(\frac{n}{\delta^2 T}\log^2 n\right)$ space. 
\end{proof}

\section*{Acknowledgements}
We thank Santosh Vempala for pointing out the connection to follow the perturbed leader. 
David P. Woodruff and Samson Zhou were supported by a Simons Investigator Award and by the NSF Grant No. CCF-1815840. 
Vaidehi Srinivas was partially supported by NSF Grant No. CCF-1652491. 
Ziyu Xu was partially supported by a PricewaterhouseCoopers Research Grant.

\bibliography{mem_experts}


\end{document}